\documentclass[11pt]{article}

\title{\Large Testable Properties in General Graphs and Random Order Streaming}
\author{Artur Czumaj
        \thanks{Department of Computer Science and Centre for Discrete Mathematics and its Applications (DIMAP), University of Warwick. Email: A.Czumaj@warwick.ac.uk. Research partially supported by the Centre for Discrete Mathematics and its Applications (DIMAP), by IBM Faculty Award, and by EPSRC award EP/N011163/1.}
    \and
Hendrik Fichtenberger
        \thanks{Department of Computer Science, TU Dortmund. Email: 	hendrik.fichtenberger@tu-dortmund.de. Research supported by ERC grant No. 307696.}
    \and
Pan Peng
        \thanks{Department of Computer Science, University of Sheffield. Email: p.peng@sheffield.ac.uk.}
    \and
Christian Sohler
        \thanks{Department of Computer Science, TU Dortmund. Email: christian.sohler@tu-dortmund.de. Research supported by ERC grant No. 307696.}
}
\date{}

\newcount\shortyear\newcount\shorthour\newcount\shortminute
\shorthour=\time\divide\shorthour by 60\shortyear=\shorthour
\multiply\shortyear by 60\shortminute=\time\advance\shortminute by
-\shortyear \shortyear=\year\advance\shortyear by -1900
\def\zeit{\number\shorthour:\ifnum\shortminute<10 0\number\shortminute
\else\number\shortminute\fi}
\date{\small \zeit{}, \today}
\date{}

\usepackage[english]{babel}

\usepackage{algorithm}
\usepackage{algpseudocode}

\usepackage{graphicx}

\usepackage{amsmath}
\usepackage{mathtools}
\usepackage{amssymb}
\usepackage{stmaryrd}

\usepackage{hyperref}

\usepackage{color,bm}
\usepackage{paralist}
\usepackage{fullpage}

\usepackage{amsthm} %

\usepackage[margin=1in]{geometry}

\usepackage[capitalise]{cleveref} %

\theoremstyle{plain}
\newtheorem{theorem}{Theorem}[section]
\newtheorem{lemma}[theorem]{Lemma}
\newtheorem{corollary}[theorem]{Corollary}
\newtheorem{claim}[theorem]{Claim}

\theoremstyle{plain}
\newtheorem{definition}[theorem]{Definition}

\def\eg/{%
	e.\,g.%
}
\def\ie/{%
	i.\,e.%
}
\def\ifft/{%
	if and only if%
}
\def\cf/{%
	cf.%
}
\def\Wlog/{%
	Without loss of generality%
}
\def\wwlog/{%
	without loss of generality%
}

\newcommand{\ent}{\mathcal{F}}

\newcommand{\degr}[1]{%
	\dg(#1)%
}

\newcommand{\ngh}[1]{%
	\Gamma(#1)%
}

\newcommand{\E}[1]{%
	E(#1)%
}

\newcommand{\defeq}{%
	:=%
}

\newcommand{\Ex}{%
	\mathrm{E}%
}

\newcommand{\setn}{%
	\mathbb{N}%
}

\newcommand{\onerange}[1]{%
	{[#1]}%
}

\newcommand{\fstop}{%
	\, .
}

\newcommand{\disk}[1]{%
	\mbox{\(#1\)-disk}%
}
\newcommand{\disks}[1]{%
	\mbox{\(#1\)-disks}%
}
\def\kdisk/{%
	\disk{k}%
}
\def\kdisks/{%
	\disks{k}%
}

\renewcommand{\ie}{i.e., }
\renewcommand{\degr}[1]{\ensuremath{\deg(#1)}}

\algtext*{EndWhile}
\algtext*{EndIf}
\algtext*{EndFor}
\newcommand{\stream}{\ensuremath{\mathcal{S}}}

\newcommand{\reach}{\ensuremath{\mathrm{Reach}}}

\renewcommand{\E}{\ensuremath{\mathrm{E}}}
\newcommand{\Var}{\ensuremath{\mathrm{Var}}}

\newcommand{\cst}{\ensuremath{{c_*}}}

\begin{document}

\begin{titlepage}
\maketitle\thispagestyle{empty}
	
\begin{abstract}
We present a novel framework closely linking the areas of \emph{property testing} and \emph{data streaming algorithms} in the setting of \emph{general graphs}. It has been recently shown (Monemizadeh et al.\ 2017) that for \emph{bounded-degree graphs}, any constant-query tester can be emulated in the random order streaming model by a streaming algorithm that uses only space required to store a constant number of words. However, in a more natural setting of \emph{general graphs}, with no restriction on the maximum degree, no such results were known because of our lack of understanding of constant-query testers in general graphs and lack of techniques to appropriately emulate in the streaming setting off-line algorithms allowing many high-degree vertices.
		
In this work we advance our understanding on both of these challenges.
\begin{itemize}
\item First, we provide \emph{canonical testers} for all constant-query testers %
    for general graphs, both, for one-sided and two-sided errors. Such canonizations were only known before (in the adjacency matrix model) for dense graphs~\cite{GT03:three} and (in the adjacency list model) for bounded degree (di-)graphs~\cite{GR11:proximity,CPS16:testing}. %
\item Using the concept of canonical testers, we then prove that every property of general graphs that is constant-query testable with one-sided error %
    can also be \emph{tested in constant-space} with one-sided error in the \emph{random order streaming model}.
\end{itemize}
Our results imply, among others, that properties like $(s,t)$ disconnectivity, $k$-path-freeness, etc. are constant-space testable in random order streams.%
\end{abstract}
\end{titlepage}

\section{Introduction}
\label{sec:introduction}

\paragraph{Graph streaming algorithms.}
One important way of processing large graphs in modern data analysis is to design \emph{graph streaming algorithms} (see, e.g., \cite{McG14:stream,Muthu05}). A graph streaming algorithm obtains the input graph as a stream of edges in some order and its goal is to process and analyze the input stream in order to compute some basic characteristics about the input graph. For example, we want to know whether the graph is connected, or bipartite, or to know its approximate maximum matching size. Following the mainstream research in data streaming, we focus on algorithms that make only a \emph{single pass} over the graph stream. Since in the single pass model every edge is seen only once, the central complexity measure of data streaming algorithms is the amount of space used to store information about the graph, with the golden standard in streaming being \emph{sublinear space}. Unfortunately, it is known that for many natural graph problems sublinear space $o(n)$ is not possible when the edges are arriving in a single pass and in arbitrary order, where $n$ is the number of vertices of the input graph \cite{HRR99:stream}.

There have been several approaches to cope with this inherent limitation of the streaming setting for graph problems. While some of the early works in graph streaming algorithms approached this challenge by allowing more than one pass over the input, the single-pass model is still considered to be the most interesting and the most natural scenario for streaming algorithms. The $\Omega(n)$ space lower bound (e.g., for testing if the graph is connected or estimating the size of transitive closure~\cite{HRR99:stream}) led to a significant number of papers designing semi-streaming algorithms, %
which are algorithms using $O(n \ \mathrm{polylog}(n))$ space, so only slightly larger than linear in the number of vertices~(see the survey \cite{McG14:stream}). While this model leads to sublinear algorithms for dense graphs, where $m$, the number of edges, is $\omega(n \ \mathrm{polylog}(n))$, for the very natural setting of \emph{sparse graphs}, semi-streaming algorithms are useless, since with $O(n \ \text{polylog}(n))$ space one can store the entire input graph (all arriving edges), and so one can trivially solve any graph problem.

Another, central approach to address the linear space lower bounds for graph streaming problems that recently received increasing attention is the \emph{random order streaming model}, where the edges arrive in random order, \ie in the order of a uniformly random permutation of the edges (see, e.g., \cite{CCM08:lower,KKS14:matching,KMM12:matching,McG14:stream,MMPS17,PS18}). The assumption about uniformly random or near-uniformly random ordering is very natural and can arise in many contexts. One might also use the random order streaming model to justify the success of some heuristics in practice, even though there exists strong space lower bound for (the worst case of) the problem. Furthermore, some recent advances have shown that some problems that are hard for adversarial streams can be solved with small space in the random order model. For example, Konrad et al.\ \cite{KMM12:matching} gave single-pass semi-streaming algorithms for maximum matching for bipartite and general graphs with approximation ratio strictly larger than $\frac12$ in the random order semi-streaming model, where it is not known if such approximation is possible in the adversarial order model. Kapralov et al.\ \cite{KKS14:matching} gave a poly-logarithmic approximation algorithm in polylogarithmic space for estimating the size of maximum matching of an unweighted graph in one pass over a random order stream, which is impossible in the adversarial order model \cite{AKL17:estimating}. Finally, \cite{PS18} showed that in the random order streaming model, even with constant space, one can approximate the number of connected components of the input graph to within an additive error of $\varepsilon n$, the size of a maximum independent set in planar graphs to within a multiplicative factor of $1 + \varepsilon$, and the weight of a minimum spanning tree of a connected input graph with small integer edge weights to within a multiplicative factor of $1 + \varepsilon$. While these results demonstrate the strength of the random order streaming model, Chakrabarti et al.\ \cite{CCM08:lower} proved that $\Omega(n)$ space is needed for any single pass algorithm for graph connectivity in the random order streaming model, almost matching the optimal $\Omega(n \log n)$ space lower bound in the adversarial order model~\cite{SW15:tight}. Furthermore, there exists a $n^{1-O(\varepsilon)}$ space lower bound for approximating the number of connected components of additive error $\varepsilon n$ adversarial order streams~\cite{HP16:stream}. This poses a central open question in the area of graph streaming algorithms, of \emph{characterizing graph problems which can be solved with small, sublinear space in the random order streaming model}.

The main goal of our paper is to address this task and to enlarge the class of graph problems known to be solvable with \emph{small space} in the \emph{random order streaming} model in a \emph{single pass}. Our main focus is on the most challenging scenario: of achieving \emph{constant space}\footnote{\label{footnote-space}Throughout the entire paper, we will count the size of the \emph{space in words} (assuming that a single word can store any single ID of a vertex or of an edge), \ie space bounds have to be multiplied by $O(\log n)$ to obtain the number of bits used. With this in mind, we use term \emph{constant space} to denote space required to store a constant number of words, or IDs, that is, $O(\log n)$ bits.}.

\paragraph{Property testing.}
A fundamental task in the study of big networks/graphs is to efficiently analyze their structural properties. For example, we may want to know if a graph is well-connected, has many natural clusters, has many copies (instances) of some specific sub-structures, etc. Given that modern networks are large, often consisting of millions and billions of nodes (web graph, social networks, etc.), the task of analyzing their structure has become recently more and more challenging, and the running-time efficiency of this task is becoming of critical importance. The framework of \emph{property testing} has been developed to address some of these challenges, aiming to trade the efficiency with the accuracy of the output, with the goal of achieving very fast algorithms.

In (graph) property testing, a tester has query access to a graph (e.g., random vertices or neighbors of a vertex for graphs), and its goal is to determine if the graph satisfies a certain property (e.g., is well-clusterable) or is far from having such a property (e.g., is ``far'' from any graph being well-clusterable; see, e.g., \cite{Gol17,GGR98:testing,GR02:testing,RS96:robust}). To be precise, we define testers as follows. Given a property $\Pi$, a tester for $\Pi$ is a (possibly randomized) algorithm that is given a proximity parameter $\varepsilon$ and oracle access to the input graph $G$. If $G$ satisfies property $\Pi$, then the algorithm must accept with probability at least $\frac23$. If $G$ is $\varepsilon$-far from $\Pi$, then the algorithm must reject with probability at least $\frac23$. If the algorithm is allowed to make an error in both cases, we say it is a \emph{two-sided error tester}; if, on the contrary, the algorithm always gives the correct answer when $G$ satisfies the property, we say it is a \emph{one-sided error tester}. Further details of the model depend on the data representation. In the main model considered in this paper, \emph{property testing for general graphs}, we will consider the \emph{random neighbor oracle access} to the input graph (cf. \cref{def:random_neighbor}), which allows to query a random neighbor of any given vertex\footnote{Our model is in contrast with the other two widely used property testing models for graphs with arbitrarily large maximum degree: In the \emph{adjacency list model}~\cite{PR02,KY14:forest}, the algorithm can perform both \emph{neighbor queries} (i.e., for the $i$-th neighbor of any vertex $v$ such that $i\le \deg(v)$), and the \emph{degree queries} (i.e., for the degree $\deg(v)$ of any vertex $v$); In the \emph{general graph model}, the algorithm is allowed to perform \emph{vertex-pair queries} (i.e., for the existence of an edge between any two vertex pair $u,v$), in addition to neighbor and degree queries~\cite{KKR04,AKKR08,Gol17}. Still, we believe that the \emph{random neighbor oracle model} considered in this paper is the most natural model of computations in the property testing framework in the context of very fast algorithms, especially those performing $O(1)$ queries. We note however, that our analysis can be generalized to other models of general graphs, for example, see \cref{sec:other-models}.}. In our model, we will say that $G$ is $\varepsilon$-far from a property $\Pi$ if any graph that satisfies $\Pi$ differs from $G$ on at least $\varepsilon |E(G)|$ edges. To analyze the performance of a tester, we will measure its quality in term of its \emph{query complexity}, which is the number of oracle queries it makes.

In the past a large body of research has focused on the analysis of various graph properties in different graph models, for example, leading to a precise characterization of all properties that can be tested with constant query complexity \cite{AFNS09,AS08} in the so-called dense model (graphs with $\Theta(n^2)$ edges), and some partial results for bounded-degree graph models (see, e.g., \cite{BSS10:minor,CSS09:hereditary,FPS19:testable,Gol17,GR02:testing,GR11:proximity,NS13:hyperfinite}). However, our understanding of the model of general graphs, graphs where each vertex can have arbitrary degree, is still rather limited. We have seen some major advances in testing graph properties for general graphs, including the results of Parnas and Ron \cite{PR02}, Kaufman et al.\ \cite{KKR04}, Alon et al.\ \cite{AKKR08}, Czumaj et al.\ \cite{CMOS11}. %
The main challenge of the study in the model of general graphs is a lack of good characterization of testable properties and of a good algorithmic toolbox for the problems in this model.
However, the importance of the general graph model and lack of major advances have been widely acknowledged in the property testing community. For example, it is recognized that the general graph model is ``most relevant to computer science applications'' and ``designing testers in this model requires the development of algorithmic techniques that may be applicable also in other areas of algorithmic research'' (see \cite[Chapter~10.5.3]{Gol17}).

\subsection{Basic Definitions and Overview of Our Results}

In this paper, we extend the approach recently introduced by Monemizadeh et al.\ \cite{MMPS17} (see also \cite{PS18}) to demonstrate a \emph{close connection between streaming algorithms and property testing} in the most general setting of \emph{general graphs}. Monemizadeh et al.\ \cite{MMPS17} show that for \emph{bounded degree graphs}, any graph property that is constant-query testable property in the adjacency list model can be tested with constant space in a single pass in random order streams.

As we mentioned, the query oracle model for general graphs we are considering is the \emph{random neighbor model}, which allows the algorithm to query a random neighbor of any specified vertex (cf. \cref{def:random_neighbor}). We have the following definitions of testing graph properties.

\begin{definition}[\textbf{Property testers in random neighbor model}]
\label{def:property-testers}
Let $\Pi = (\Pi_n)_{n \in \setn}$ be a graph property, where $\Pi_n$ is a property of graph of $n$ vertices. We say that \emph{$\Pi$ is testable with query complexity $q$}, if for every $\varepsilon$ and $n$, there exists an algorithm (called \emph{tester}) that makes at most $q = q(n, \varepsilon)$ oracle queries, and with probability at least $\frac23$, accepts any $n$-vertex graph satisfying $\Pi$, and rejects any $n$-vertex graph that is $\varepsilon$-far from satisfying $\Pi$. If $q = q(\varepsilon)$ is a function independent of $n$, then we call $\Pi$ \emph{constant-query testable}. If the tester always accepts graphs that satisfy $\Pi$, we say that it has \emph{one-sided error}. Otherwise, we say the tester has \emph{two-sided error}.
\end{definition}

We notice that the definition above is generic and can be applied to any of the query oracle models (see e.g.~\cite{Gol17}). However, since our main query oracle model is the random neighbor model, only for that model we will use the terminology from \cref{def:property-testers} without a direct reference to the query oracle model. We first give canonical testers in this model. In order to do so, we introduce a process called \emph{$q$-random BFS} ($q$-RBFS) starting with any specified vertex $v$, \ie a BFS of depth $q$ that is restricted to visiting at most $q$ random neighbors for every vertex (see \cref{def:randm-bfs}). We call the subgraph obtained by a $q$-RBFS a \emph{$q$-bounded disc}. Our first result is informally stated as follows.

\begin{theorem}[informal; cf. \cref{thm:canonical-tester}]\label{thm:informal_cano}
If a property $\Pi=(\Pi_n)_{n \in \setn}$ is testable with $q=q(\varepsilon)$ queries in the random neighbor model, then it can also be tested by a canonical tester that
\begin{enumerate}
\item samples $q'$ vertices;
\item performs $q'$-RBFS from each sampled vertex;
\item accepts if and only if the explored subgraph does not contain any (forbidden) graph $F\in \mathcal{F}$,
\end{enumerate}
where $q'$ depends only on $q$, and $\mathcal{F}$ is a family of rooted graphs such that each graph $F \in \mathcal{F}_{n}$ is the union of $q'$ many $q'$-RBFS bounded discs.
\end{theorem}

We remark that similar canonical testers have been given for dense graphs \cite{GT03:three}, bounded degree graphs and digraphs \cite{GR11:proximity,CPS16:testing}. Actually, our proof for the above theorem heavily builds upon \cite{GR11:proximity,CPS16:testing}, though our analysis requires some extensions to deal with general graphs (of possibly unbounded degree). To formally state our result regarding testing graph properties in streaming, we introduce the following definition.

\begin{definition}[\textbf{Property testers in the streaming model}]
Let $\Pi = (\Pi_n)_{n\in \setn}$ be a graph property, where $\Pi_n$ is a property of graph of $n$ vertices. We say that $\Pi$ is testable with space complexity $q$, if for every $\varepsilon$ and $n$, there exists an algorithm that performs a single pass over an edge stream of an $n$-vertex graph $G$, uses $q = q(n,\varepsilon)$ \emph{words} of space, and with probability at least $\frac23$, accepts $G$ satisfying $\Pi$, and rejects $G$ that is $\varepsilon$-far from satisfying $\Pi$. If $q = q(\varepsilon)$ is a function independent of $n$, then we call $\Pi$ \emph{constant-space testable}. If the tester always accepts the property, then we say that the property can be tested with \emph{one-sided error}. Otherwise, we say the tester has \emph{two-sided error}.
\end{definition}

Our main result and our main technical contribution is the \emph{transformation} of a one-sided error property tester in the random-neighbor model with constant \emph{query} complexity into a one-sided error property tester in the streaming model with constant \emph{space} complexity.

\begin{theorem}\label{thm:stream-test}
Every graph property $\Pi$ that is constant-query testable with one-sided error in the random neighbor model is also constant-space testable (space measured in words) with one-sided error in the random order graph streams.
\end{theorem}

\paragraph{Applications.}
We believe that the main contribution of our paper is a general transformation presented in \cref{thm:stream-test}. However, we admit that the number of properties testable with one-sided error with a constant number of queries in the random neighbor model is rather limited. Still, we can apply our transformation to, for example, the property of being $(s,t)$-disconnected (i.e., there is no path between $s$ and $t$), see, e.g., \cite{YK12:disconnectivity}\footnote{The constant-query tester from \cite{YK12:disconnectivity} performs degree queries and neighbor queries, but it is straightforward to simulate it in the random neighbor model. Indeed, the algorithm in \cite{YK12:disconnectivity} only needs to repeatedly perform a constant-length random walks from $s$ and reject if only if one path from $s$ to $t$ is found. Such an algorithm can be trivially simulated in the random neighbor model.}. Furthermore, our transformation actually holds for \emph{any} restricted class of graphs (a promise on the structure of the input graph), including the class of planar, minor-free graphs, or the class of bounded degree graphs. Since bipartiteness in planar graphs (or minor-free graphs) is testable in the random neighbor model~\cite{CMOS11}, it is also one-sided error testable in random order streams; the same holds for testing $H$-freeness in planar or minor-free graphs \cite{CS19}. Furthermore, our techniques can also be used to transform any constant-query tester (with one-sided error) in the \emph{random neighbor/edge model} (cf. \cref{sec:other-models}) to the random order streaming model, where the random neighbor/edge model allows to sample an edge uniformly at random. Therefore, for example, since the property of being $P_k$-free (there is no path of length $k$) is constant-query testable in the random neighbor/edge model with one-sided error \cite{IY18:parameterized}, $P_k$-freeness is also constant-space testable with one-sided error in the random order graph streams. Similarly, it is not hard to see that the property of being $d$-bounded (%
the maximum degree is at most $d$) is constant-query testable in the random neighbor/edge model\footnote{If $G$ is $\varepsilon$-far from the property, then at least $\Omega(\varepsilon |E|)$ edges are incident to a node with degree at least $d+1$. Thus, we can simply sample a constant number of edges and check if either of its endpoints has degree at least $d+1$.}, and therefore this property too is constant-space testable with one-sided error in the random order graph streams.

\subsection{Challenges and Techniques}

The result about constant-space streaming algorithms for bounded-degree graphs by Monemizadeh et al.\ \cite{MMPS17} is obtained by noting that any constant-query complexity tester basically estimates the distribution of local neighborhoods of the vertices (see, e.g., \cite{CPS16:testing,Gol17,GR11:proximity}) and emulating any such algorithm on a random order graph stream using constant space. Unfortunately, this approach inherently relies on the assumption that the input graph is of bounded degree. This limitation comes from two ends: on one hand, there has not been known any versatile description of testers for constant-query testable graph properties of general graphs, and on the other hand, the streaming approach from \cite{MMPS17} relies on a breadth-first-search-like graph exploration that is possible (with constant space) only when the input graph has no high-degree vertices. A follow-up paper \cite{PS18} made the first attempt to address the challenge of dealing with general degrees, and considered some problems in which one can \emph{ignore} high degree vertices (e.g., for approximating the number of connected components or the size of a maximum independent set in planar graphs).

One important reason why the earlier approaches have been failing for the model of general graphs, without bounded-degree assumption, was our lack of understanding of constant-query time testers in general graphs and the lack of techniques to appropriately emulate off-line algorithms allowing many high-degree vertices. In this paper, we advance our understanding on both of these challenges.

\paragraph{A general and simple canonical tester.}
To derive a canonical tester for constant-query testable properties in the random neighbor model, we introduce the process \emph{$q$-random BFS} ($q$-RBFS): it starts from any specified vertex $v$, and then performs a BFS-like exploration of depth~$q$ that is restricted to visiting at most $q$ random neighbors at each step (see \cref{def:randm-bfs} for the formal definition). We call the subgraph obtained by a $q$-RBFS a \emph{$q$-bounded disc}. With the notion of $q$-RBFS and $q$-bounded discs, we are able to transform every constant-query tester for properties of general graphs into a \emph{canonical tester} that works as follows: it samples $q$ random vertices, performs a $q$-RBFS from each sampled vertex, and rejects \ifft/ the (non-induced) subgraph it has seen (which is a union of $q$-bounded discs) is isomorphic to some member of a family $\mathcal{F}$ of forbidden subgraphs (see \cref{thm:informal_cano,thm:canonical-tester}). Furthermore, such a canonical tester preserves one-sided error, while the query complexity blows up exponentially.

Canonical testers provide us a systematic view of the behavior of constant-query testers in the random neighbor model. They further tells us that in order to test a constant-query testable property $\Pi$, it suffices to estimate the probability that some forbidden subgraph in $\mathcal{F}$ is found by a $q$-RBFS starting from a randomly sampled vertex. Slightly more formally, we define the \emph{reach probability} of a subgraph $F\in\mathcal{F}$ to be the probability that a $q$-RBFS starting from a uniformly chosen vertex $v$ sees a graph that is isomorphic to $F$. In particular, if we can estimate these reach probabilities in random order streams, then we can also test $\Pi$ accordingly.

The problem with this approach is that it is hard to estimate the reach probabilities of subgraphs in $\mathcal{F}$. The main challenge here is that a forbidden subgraph $F \in \mathcal{F}_n$ may be the union of more than two or more subgraphs obtained from different $q$-RBFS that may intersect with each other.

\paragraph{A refined canonical tester.}
To cope with the challenge mentioned above of estimating the reach probabilities of subgraphs in $\mathcal{F}$, we decompose each forbidden subgraph $F \in \mathcal{F}_n$ into all possible sets of intersecting $q$-bounded discs whose union is $F$ and then try to recover $F$ from these sets. In order to recover $F$ from such a decomposition, we have to identify and monitor \emph{vertices that are contained in more than one $q$-bounded disc of $F$}.

We refine the analysis of the canonical tester and separate the $q$-bounded discs explored by each $q$-RBFS and keep track of their intersections (cf. \cref{thm:refined_characterization}). We first observe that for every input graph $G$ and every $\varepsilon$, there exists a \emph{small fixed set} $V_\alpha \subseteq V$ of all vertices whose probability to be visited by a random $q$-RBFS from a random vertex exceeds some small threshold $\alpha$ (depending on $q$ and $\varepsilon$, but independent of $n$). In other words, with constant probability, the subgraphs explored by multiple $q$-RBFS in the canonical tester will only overlap on vertices from $V_\alpha$. Furthermore, we prove that the degree of all vertices in $V_\alpha$ is at least linear (in $n$), and with constant probability, two random $q$-RBFS subgraphs will not share any edge. Since $V_\alpha$ has constant size, each $q$-bounded disc can be viewed as a \emph{colored} $q$-bounded disc \emph{type} such that each vertex in $V_\alpha$ is assigned a unique color from a constant-size palette. This way, it is possible to reversibly decompose each $F\in \mathcal{F}_n$ into a multiset of colored $q$-bounded disc types (actually, there may be many such multisets for each $F$): since the $q$-bounded discs that are explored by different $q$-RBFS intersect only at vertices in $V_\alpha$, $F$ is obtained by identifying vertices of the same color. See \cref{fig:stitching} for an example.

\begin{figure}
\centering
\begingroup%
  \makeatletter%
  \providecommand\color[2][]{%
    \errmessage{(Inkscape) Color is used for the text in Inkscape, but the package 'color.sty' is not loaded}%
    \renewcommand\color[2][]{}%
  }%
  \providecommand\transparent[1]{%
    \errmessage{(Inkscape) Transparency is used (non-zero) for the text in Inkscape, but the package 'transparent.sty' is not loaded}%
    \renewcommand\transparent[1]{}%
  }%
  \providecommand\rotatebox[2]{#2}%
  \newcommand*\fsize{\dimexpr\f@size pt\relax}%
  \newcommand*\lineheight[1]{\fontsize{\fsize}{#1\fsize}\selectfont}%
  \ifx\svgwidth\undefined%
    \setlength{\unitlength}{227.96786571bp}%
    \ifx\svgscale\undefined%
      \relax%
    \else%
      \setlength{\unitlength}{\unitlength * \real{\svgscale}}%
    \fi%
  \else%
    \setlength{\unitlength}{\svgwidth}%
  \fi%
  \global\let\svgwidth\undefined%
  \global\let\svgscale\undefined%
  \makeatother%
  \begin{picture}(1,0.56233427)%
    \lineheight{1}%
    \setlength\tabcolsep{0pt}%
    \put(0.02350654,0.31887553){\color[rgb]{0,0,0}\makebox(0,0)[lt]{\lineheight{0.5}\smash{\begin{tabular}[t]{l}u\end{tabular}}}}%
    \put(0.27612672,0.35764978){\color[rgb]{0,0,0}\makebox(0,0)[lt]{\lineheight{0.5}\smash{\begin{tabular}[t]{l}v\end{tabular}}}}%
    \put(0,0){\includegraphics[width=\unitlength,page=1]{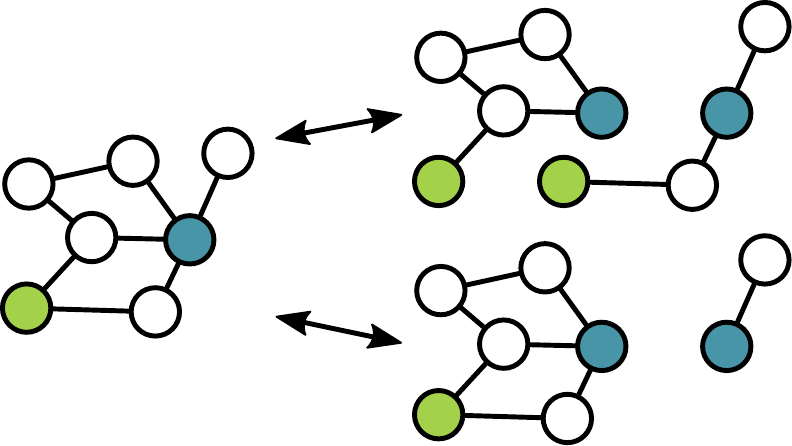}}%
  \end{picture}%
\endgroup%

\caption{\small Consider the graph on the left, which can be decomposed into colored $3$-bounded disc types (which are rooted at $u$ and $v$ in this example) in more than one way. However, it is always possible to recover the original graph by identifying vertices of the same color. Furthermore, every mapping is bijective because every color is assigned at most once per disc. If the colored vertices correspond to the vertices in $V_\alpha$, every forbidden graph $F \in \mathcal{F}_n$ from \cref{thm:canonical-tester} corresponds to a decomposition into edge-disjoint colored $q$-bounded discs $F' \in \mathcal{F}'_n$ in \cref{thm:refined_characterization}, which intersect only at colored vertices.}
\label{fig:stitching}
\end{figure}

These properties are crucial to describe the forbidden subgraphs in terms of the graphs seen by the $q$ many $q$-RBFS that the canonical tester performs and a \emph{constant-size} description of their interaction, \ie how they overlap.

\paragraph{Simulation in the streaming.}
In the streaming, in order to simulate $q$-RBFS, it is natural to consider the following procedure called \textsc{StreamCollect} ($q$-SC, see \cref{alg:collect-nbh-stream}) to explore the subgraph surrounding any specified vertex. That is, it maintains a connected component $C$ that initially contains only the start vertex. Whenever it reads an edge that connects to the current $C$ and the augmented component may be observed by a run of $q$-RBFS, it adds the edge to $C$.

Note that one important feature of random order streams is that we would see the right exploration (as in the query model) with constant probability, while it is challenging to verify if the subgraph we collected from the stream is indeed the right exploration (cf. \cite{MMPS17,PS18} for a more detailed discussion). In our setting, as we mentioned, another technical difficulty is to analyze whether subgraphs found by running the stream procedure multiple times \emph{intersect} in exactly the same way as the $q$-bounded discs that are found by $q$-RBFS.

With the refined canonical tester, which specifies how different $q$-RBFS procedures intersect, we are able to \emph{simulate one-sided error constant-query testers} in the random neighbor model for general graphs in the \emph{random order streaming model}. Since the considered property $\Pi$ is one-sided error testable in the random neighbor model, it suffices to detect a forbidden subgraph $F$ in the family $\mathcal{F}$ corresponding to $\Pi$ with constant probability. That is, it suffices to show that if the graph is far from having the property, then for any forbidden subgraph $H$ that can be reached by the canonical tester with probability $p$, it can also be detected by \emph{multiple} \textsc{StreamCollect} subroutines with probability at least $cp$ for some suitable constant $c$.\footnote{Note that this is not sufficient for simulating two-sided error testers. Let us take the property connectivity (which is $2$-sided error testable in random neighbor model) for example. If the input graph is a path on $n$ vertices, then a $q$-RBFS will detect a forbidden subgraph (i.e., a path of constant length that is not connected to the rest) corresponding to connectivity with small constant probability, while a $q$-SC might see a forbidden subgraph with high constant probability. That is, in order to test connectivity, we need to be able to approximate the \emph{frequencies} of the forbidden subgraphs, for which our current techniques fail.	
} %

In order to do so, we first decompose the forbidden subgraphs that characterize the property into colored subgraphs, where each subgraph corresponds to a run of $q$-RBFS and vertices in $V_\alpha$ are colored with a unique color. Then, we prove that for a sufficiently large sample of vertices, the $q$-SC subroutines starting from these sampled vertices will collect, for each colored subgraph $H$, at least as many instances of $H$ as the canonical property tester sees. Suppose that the input graph is far from the property. Since the subgraphs observed by the canonical tester intersect only at vertices in $V_\alpha$, \ie colored vertices, with constant probability, it is possible to stitch a forbidden subgraph by identifying vertices of the same color in the analysis.

The analysis of this procedure is two-fold. First, we show that if a single run of $q$-RBFS from $v$ sees a certain colored $q$-bounded disc type with probability $p$ (where the colored vertices are $V_\alpha$), then a single run of $q$-SC from $v$ sees this disc type with probability $cp$ for some suitable constant $c$ (see \cref{lemma:streamproblower}).

The second step (which is the main technical part) is to show that if the probability that a $q$-RBFS from a random vertex sees a colored $q$-bounded disc type $\Delta$ is $p$, then with constant probability, for a sufficiently large sample set $S$, the calls to $q$-SC from vertices in $S$ will also see a $q$-bounded disc type $\Delta$, even though there are intersections from different $q$-SCs (see \cref{lemma:lower_bound_disc_probability}). Then we can show that if the input graph is far from the property, with constant probability, we can stitch the colored $q$-bounded discs to obtain a forbidden subgraph $F\in\mathcal{F}$ (see \cref{thm:stream-test}).

Finally, we remark that colors are only used in the analysis as the streaming algorithm can identify intersections of multiple $q$-SC by the vertex labels. However, the colors are crucial to the analysis: without colors, we cannot guarantee that the $q$-bounded disc types found by multiple $q$-SCs can be stitched in the same way as the $q$-bounded disc types found by $q$-RBFS. Here is an example: Consider some constant-query testable property $\Pi$ such that the set of forbidden subgraphs $\mathcal{F}$ contains a graph $F$ that is not a subgraph of any single $q$-bounded disc type (i.e, it is the union of at least two intersecting $q$-bounded disc types). For the sake of illustration, a concrete example %
is provided in \cref{fig:whycolors}. In order to reject, the canonical property tester needs to find at least two intersecting $q$-bounded discs such that their union contains $F$ as a subgraph. However, even if we bound, for each \emph{uncolored} $q$-bounded disc type $\Delta$, the probability that $q$-SC finds $\Delta$ by some constant fraction of the probability that $q$-RBFS finds $\Delta$, this is not sufficient to conclude that the probability that multiple $q$-SCs find a copy of $F$ is bounded by a constant fraction of the probability that multiple $q$-RBFS find a copy of $F$. The reason is that $q$-SC might only find copies of $\Delta$ that are not intersecting, while $q$-RBFS might tend to find copies of $\Delta$ that intersect. Again, see \cref{fig:whycolors} for an example. Therefore, we need to preserve, for each $q$-bounded disc type $\Delta$, the information which of the corresponding vertices in the input graph are likely to be contained in more than one $q$-RBFS for the analysis.

\begin{figure}
\centering
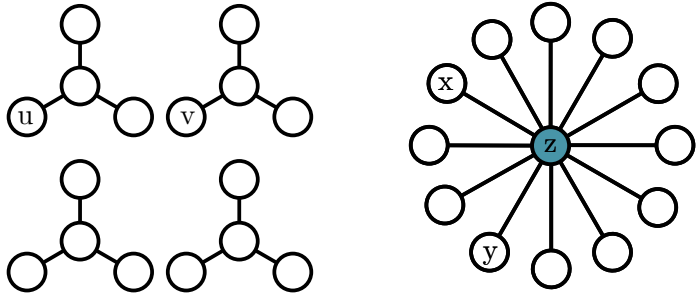
\caption{\small The above graph, which is composed of $3$-stars and a $\omega(1)$-star with root $z$ and which should be thought of as a subgraph of some larger graph, illustrates the need for colors in our analysis of the streaming property tester. Although the $2$-bounded discs of $u$, $v$ $x$ and $y$ are all $3$-stars (with constant probability over the randomness of the neighbor queries), exploring $u$ and $v$ by $q$-RBFS does not result in finding a $6$-star, while it is likely to find a $6$-star by exploring $x$ and $y$. Even if we prove that the probability that a $q$-SC finds uncolored $3$-stars is lower bounded by some constant fraction of the probability that $q$-RBFS finds uncolored $3$-stars, we still cannot rule out that $q$-SC might tend to find leaves of the small stars (like $u$ and $v$) while $q$-RBFS tends to find leaves of the big star (like $x$ and $y$). Observe that here, $z$ is the only vertex that is likely contained in two different $q$-RBFS due to its high degree.}
\label{fig:whycolors}
\end{figure}

\section{Preliminaries}
\label{sec:prelim}

Let $G = (V, E)$ be an undirected graph. We will assume that the vertex set $V$ of $G$ is $[n] = \{1, \dots, n \}$, and we let $\degr{v}$ denote the degree of $v \in V$. Sometimes, we use $V(G)$ to denote the vertex set $V$ of $G$ and $E(G)$ to denote the edge set $E$ of $G$. We let $\stream(G)$ denote the input stream of edges that defines $G$. In this paper, we consider streaming algorithms for random order streams, \ie the input stream $\stream(G)$ to our algorithm is drawn uniformly from the set of all permutations of $E$. We are interested in \emph{streaming algorithms that have constant space complexity} in the size of the graph, where we count the size of the space in words, \ie space bounds have to be multiplied by $O(\log n)$ to obtain the number of bits used, see also footnote~\ref{footnote-space}.

A graph $G$ is called a \emph{rooted} graph if at least one vertex in $G$ is marked as \emph{root}. Let us define the notion of a root-preserving isomorphism.

\begin{definition}%
\label{def:root-preserving-isomorphism}
Given two rooted graphs $H_1$ and $H_2$, a \emph{root-preserving isomorphism} from $H_1$ to $H_2$ is a bijection $f: V(H_1) \rightarrow V(H_2)$ such that
\begin{inparaenum}[(i)]
\item if $u$ is the root of $V(H_1)$ then $f(u)$ is the root of $V(H_2)$, and
\item that $(u,v)\in E(H_1)$ if and only if $(f(u),f(v))\in E(H_2)$.
\end{inparaenum}
If there is a root-preserving isomorphism from $H_1$ to $H_2$ then we say that $H_1$ is root-preserving isomorphic to $H_2$ and denote it by $H_1\simeq H_2$.
\end{definition}

\section{Canonical Constant-Query Testers in General Graphs}
\label{sec:canonincal-testers}

In this section, we present our main result on the \emph{canonical testers for constant-query testable properties in general graphs}. After starting with some basic definitions, we will present two canonical testers for constant-query testable properties in general graphs. Our first canonical tester is of a general form (see \cref{subsec:canonical-tester-general}) and our second tester (see \cref{thm:refined_characterization} in \cref{subsec:canonical-tester-refined}) is slightly more refined, allowing for a more natural use later in the setting of streaming algorithms in \cref{sec:testing-in-streams}.

We note that in this paper we focus on one specific model of access to the input graph, the \emph{random neighbor model}. It is possible to extend some of our analysis (of canonical testers) to some other graph access models, though (cf. \cref{sec:other-models}).

\subsection{Random BFS and Bounded Discs}

\paragraph{Property testing in query oracle model.}
Since we consider general graphs, without any bounds for vertex degrees, we have to carefully define the access provided to the input graph in the property testing framework. The access to the input graph is given by \emph{queries} to an \emph{oracle} representing the graph. There have been several oracles considered in the literature for general graphs, but our main focus is on the \emph{random neighbor model}, which we consider to be natural for graphs with unbounded degree, especially in the context of properties testable with a constant number of queries.

\begin{definition}[\textbf{Random neighbor model}]
\label{def:random_neighbor}
In the \emph{random neighbor model}, an algorithm is given $n \in \setn$ and access to an input graph $G=(V,E)$ by a query oracle, where $V = [n]$. The algorithm may ask queries based on the entire knowledge it has gained by the answers to previous queries. The \emph{random neighbor query} specifies a vertex $v \in V$ and the oracle returns a vertex that is chosen i.u.r. (independently and uniformly at random) from the set of all neighbors of $v$.
\end{definition}

Notice that in the random neighbor model, since $V = [n]$, the algorithm can also trivially select a vertex from $V$ i.u.r. We believe that the random neighbor model is the most natural model of computations in the property testing framework in the context of very fast algorithms (especially those of constant query complexity), and therefore our main focus is on that model. However, we want to point out that some of our results are sufficiently general to apply to a larger variety of the query oracle models, though we will not elaborate about it here (cf. \cref{sec:other-models}).

We describe the first canonical testers of all constant-query testers (in the random neighbor model) for general graphs, both, for one-sided and two-sided errors. With this canonization, we can model all graph properties testable with a constant number of queries using so-called \emph{canonical testers}; see \cref{thm:canonical-tester,thm:refined_characterization} for formal statements.%

To formalize our canonical testers for all constant-query testers in the random neighbor model, we will use the following two definitions of constrained random BFS-like graph exploration and of bounded discs.

We begin with the definition of a $q$-RBFS process, which starts at some vertex and explores its neighborhood in a BFS-like fashion, conditioned on a bound of the depth and the breadth of the exploration (see \cref{def:randm-bfs} for formal definition and \cref{alg:random-bfs} for the detailed implementation).

\begin{definition}[\textbf{$q$-random BFS}]
\label{def:randm-bfs}
Let $q > 0$ be an integer and $G$ be a simple graph. For any vertex $v \in V(G)$, the \emph{$q$-random BFS} (abbreviated as \emph{$q$-RBFS}) explores a random subset of the $q$-neighborhood of $v$ in $G$ iteratively as follows. First, it initializes a queue $Q = \{ v \}$ and a graph $H = (\{ v \}, \emptyset)$. Then, in every iteration, it pops a vertex $u$ from $Q$ and samples $q$ random neighbors $s_{u,1}, \ldots, s_{u,q}$ of $u$. For every edge $e = \{ u, s_{u,i} \}$, it adds $s_{u,i}$ and the directed edge $(u, s_{u,i})$ to $H$. Furthermore, if $s_{u,i}$ has distance less than $q$ from $v$ in $H$ and $s_{u,i}$ has not been added to $Q$ before, $s_{u,j}$ is appended to $Q$. When $Q$ is empty, all edges in $H$ are made undirected (without creating parallel edges) and $H$ is returned.
\end{definition}

\begin{algorithm}[!h]
	\caption{$q$-random BFS}
	\label{alg:random-bfs}
	\begin{algorithmic}
		\Function{RandomBFS}{$G, v, q${}}
		\State $Q \gets$ empty queue; $\textrm{enqueue}(Q,v)$
		\State $\forall w \in V: \ell[w] \gets \infty$
		\State $\ell[v] \gets 0$
		\State $H \gets (\{ v \}, \emptyset)$ with $v$ as root
		\While{$Q$ not empty}
		\State $u \gets$ pop element from $Q$
		\For{$1 \le i \le q$}
		\State $s_{u,i} \gets$ query oracle for random neighbor of $u$
		\State add vertex $s_{u,i}$ and edge $(u,s_{u,i})$ to $H$
		\If{$\ell[u] < q - 1 \wedge \ell[s_{u,i}] = \infty$}
		\State $\ell[s_{u,i}] \gets \ell[u] + 1$
		\State $\textrm{enqueue}(Q,s_{u,i})$
		\EndIf
		\EndFor
		\EndWhile
		\State \Return undirected $H$ without parallel edges
		\EndFunction
	\end{algorithmic}
\end{algorithm}

Any output of $q$-RBFS algorithms can be described in a static form using the concept of bounded discs.

\begin{definition}[\textbf{$q$-bounded disc}]
\label{def:bounded-disc}
For a given $q \in \setn$, graph $G = (V,E)$, and vertex $v \in V$, a \emph{$q$-bounded disc of $v$ in $G$} is any subgraph $H$ of $G$ that is rooted at $v$ and can be returned by \textsc{RandomBFS}$(G,v,q)$. In this case, vertex $v$ is called a \emph{root} of the $q$-bounded disc $H$ and the maximum distance from $v$ to any other vertex in $H$ is called the \emph{radius} of $H$.
\end{definition}

All $q$-bounded discs that are root-preserving isomorphic form an equivalence class.

\begin{definition}[\textbf{$q$-bounded disc type}]
\label{def:bounded-disc-type}
Let $H$ be a $q$-bounded disc. The equivalence class of $H$ with respect to $\simeq$, i.\ e., the existence of a root-preserving isomorphism (see \cref{def:root-preserving-isomorphism}), is called the $q$-bounded disc type of $H$.
\end{definition}

In the following \cref{subsec:canonical-tester-general}, we introduce the input model, property testing and random breadth-first search.

\subsection{Canonical Testers: A General Version}
\label{subsec:canonical-tester-general}

In the following, we present the proof of our first main result. We show that any tester with query complexity $q = q(\varepsilon,n)$ in the random neighbor model can be simulated by a \emph{canonical tester} that samples $q' = O(q)$ vertices and rejects if and only if the union of the subgraphs induced by the $q'$-RBFS from the sampled vertices belongs to some family of forbidden graphs.

\begin{theorem}[\textbf{Canonical tester}]
\label{thm:canonical-tester}
Let $\Pi = (\Pi_n)_{n \in \setn}$ be a graph property that can be tested in the random neighbor model with query complexity $q = q(\varepsilon,n)$ and error probability at most $\frac{1}{3}$. Then for every $\varepsilon$, there exists an infinite sequence $\mathcal{F} = (\mathcal{F}_{n})_{n \in \setn}$ such that for every $n \in \setn$,
\begin{itemize}
\item $\mathcal{F}_{n}$ is a set of rooted graphs such that each graph $F \in \mathcal{F}_{n}$ is the union of $q'$ many $q'$-bounded discs;
\item the property $\Pi_n$ on $n$-vertex graphs can be tested with error probability at most $\frac13$ by the following canonical tester:
\begin{enumerate}
\item sample $q'$ vertices i.u.r. and mark them roots;
\item for each sampled vertex $v$, perform a $q'$-RBFS starting at $v$;
\item reject if and only if the explored subgraph is root-preserving isomorphic to some $F \in \mathcal{F}_{n}$,
\end{enumerate}
\end{itemize}
where $q'=cq$ for some constant $c > 1$. The query complexity of the canonical tester is $q^{O(q)}$. Furthermore, if $\Pi = (\Pi_n)_{n \in \setn}$ can be tested in the random neighbor model with one-sided error, then the resulting canonical tester for $\Pi$ has one-sided error too, i.e., the tester always accepts graphs satisfying $\Pi$.
\end{theorem}

\begin{proof}
Our proof follows the approach used earlier \cite{GR11:proximity,CPS16:testing}, though our analysis requires some extensions to deal with general graphs (of possibly unbounded degree).

Let $ \mathcal{T}$ be a tester for $\Pi_n$ on $n$-vertex graphs with error probability amplified to at most $ \frac{1}{6}$. Note that the query complexity of $\mathcal{T}$ is $q' = cq$ for some constant $c > 1$. We will first convert $\mathcal{T}$ into a tester $\mathcal{T}_1$ that samples a random subgraph $H$ of the input graph and answers all of $\mathcal{T}$'s queries using this subgraph. In particular, it samples $q'$ vertices and then returns the output on the basis of the subgraphs explored by all $q'$-RBFS that start at these vertices. Then, we convert $\mathcal{T}_1$ into a tester $\mathcal{T}_2$ whose output depends only on the edges and non-edges in the explored subgraph, the ordering of all explored vertices and its own coins. Next, we convert $\mathcal{T}_2$ into a tester $\mathcal{T}_3$ whose output is independent of the ordering of all explored vertices. Therefore, after sampling $H$, the probability that $\mathcal{T}_3$ accepts the input graph is equal for all query-answer sequences $((u_i) = v_i)_{i \in \onerange{q'}}$ that $\mathcal{T}$ may ask and observe. Finally, we convert $ \mathcal{T}_3$ into a tester $\mathcal{T}_4$ that returns the output deterministically according to this unlabeled version of $H$ where roots are marked identically.

Let \textbf{$\mathcal{T}_1$} be the tester that first samples a set $S_0$ of $q'$ vertices i.u.r. and then explores a subgraph by starting a $q'$-RBFS at each of vertices. We mark all vertices in $S_0$ as roots and denote the union of all subgraphs by $H=(V',E')$. We use $\mathcal{T}_1$ to simulate the execution of $\mathcal{T}$ in following way. Given a random neighbor query to an $n$-vertex graph $G=(V,E)$, the tester $\mathcal{T}_1$ will select on-the-fly a random, uniformly distributed permutation $\pi \colon V \rightarrow V$ and provide oracle access to the permuted graph $\pi(G) = (V, \pi(E))$, where $\pi(E) \defeq \{ (\pi(u), \pi(v) \mid (u,v) \in E \}$. Initially, all vertices in $S_0$ are considered unused. In the simulation, when $\mathcal{T}$ makes a query for a random neighbor of vertex $v$ and if $v$ has not appeared in any prior query or answer, then the tester $\mathcal{T}_1$ allocates $v$ to an unused vertex $u$ in the sample set $S_0$, and we let $\pi(v) = u$ and $u$ will be considered as used; otherwise $\mathcal{T}_1$ uses the allocation $ \pi(v)$ determined in the previous steps of the tester $\mathcal{T}_1$. To answer the query for a random neighbor of $v$, $\mathcal{T}_1$ selects $w = s_{\pi(v),{j+1}}$ (see \cref{alg:random-bfs}) where $j$ is the number of random neighbor queries of $v$ that $\mathcal{T}$ has issued so far. If $w$ has been selected as the selected as the image of some vertex in the permutation $\pi$ before, then $\mathcal{T}_1$ returns $\pi^{-1}(w)$; otherwise $ \mathcal{T}_1$ returns a random unused value (vertex label) $x$ and we let $\pi(x) = w$. If $w \in S_0$, then $w$ will be considered used. The returned values will then be fed into $\mathcal{T}$. The tester $\mathcal{T}_1$ makes the same decision as the final decision of $\mathcal{T}$ after receiving all the necessary query answers. It follows that
\begin{align*}
    \Pr[\mathcal{T}_1 \text{ correctly answers } G]
        &=
    \sum_{\pi} \Pr[\mathcal{T} \text{ correctly answers } \pi(G) \mid \pi] \cdot \Pr[\pi]
        \ge
    n! \cdot \frac{5}{6n!}
        =
    \frac56
    \fstop
\end{align*}

We make $\mathcal{T}_1$ label oblivious by defining the new tester \textbf{$\mathcal{T}_2$} to be the one that accepts $G$ with the \emph{average probability} that $\mathcal{T}_1$ accepts $G$ over the choice of $\pi$, \ie
\begin{align*}
    \Pr[\mathcal{T}_2 \text{ accepts } G \mid \pi]
        &=
    \sum_{\pi'} \Pr[\mathcal{T}_1 \text{ accepts } G \mid \pi'] \cdot \Pr[\pi']
        =
    \Pr[\mathcal{T}_1 \text{ accepts } G]
    \fstop
\end{align*}
Since it suffices to consider all labellings of $H$ and random coins of $\mathcal{T}$, this does not require any additional queries.

We make $\mathcal{T}_2$ oblivious of the order of $S_0$ and all $ s_{v,i}$ by considering the uniform distribution $\mathcal{U}$ over all permutations of elements in $S_0$ and all permutations of $(s_{v,i})_{i \in \onerange{q'}}$ for all $v$. In particular, we let the resulting tester \textbf{$\mathcal{T}_3$} accept with the average probability that $\mathcal{T}_2$ accepts $G$, where the probability is taken over the choice of $S \in \mathcal{U}$, \ie
\begin{align*}
    \Pr[\mathcal{T}_3 \text{ accepts } G \mid S, \pi]
        &=
    \sum_{S \in \mathcal{D}(S_0)} \Pr[\mathcal{T}_2\text{ accepts } G \mid S, \pi] \cdot \Pr[S|\pi]
        =
    \Pr[\mathcal{T}_2 \text{ accepts } G \mid \pi]
    \fstop
\end{align*}
Again, this does not require any additional queries because we only need to consider all permutations of sampled vertices and random coins of $\mathcal{T}$.

Let \textbf{$\mathcal{T}_4$} be the tester obtained from $\mathcal{T}_3$ that accepts (with probability $1$) the input graph \ifft/ the acceptance probability associated with the explored subgraph $H$ is at least $\frac12$. Since the acceptance probability of $ \mathcal{T}_3$ does not change when vertices are relabeled or $S$ is reordered, its depends only on $H$ up to isomorphism and its internal randomness. Similarly to the proof of Lemma~4.4 in \cite{GT03:three}, we can prove that $\mathcal{T}_4$ is a tester for $\Pi_n$ with error probability~$\frac13$.

Note that the decision of $\mathcal{T}_4$ is deterministic after $H$ is determined. We define $ \mathcal{F}_{n}$ to be the set of graphs that is a union of $q'$ many $q'$-bounded discs on which which the tester rejects.

Finally, let us observe that if the original tester for $\Pi = (\Pi_n)_{n \in \setn}$ can be tested in the random neighbor model with query complexity $q'$ with one-sided error, then all steps of our simulations ensure that the resulting canonical tester has one-sided error too. In particular, we have that $\Pr[\mathcal{T}_1 \text{ correctly answers } G] = 1$ for all $G \in \Pi_n$, and all the remaining steps maintain one-sided error because they do not decrease the acceptance probability for any graph $G$ that is accepted by $\mathcal{T}$ with probability at least $\frac{2}{3} < 1$.
\end{proof}

\subsection{Canonical Testers Revisited: Identifying Vertices in the Intersecting Discs}
\label{subsec:canonical-tester-refined}

\Cref{thm:canonical-tester} provides us a canonical way of testing constant-query testable properties (in the random neighbor model) by relating the tester to a set of forbidden subgraphs $\mathcal{F}_n$ for every $n \in \mathbb{N}$. However, as we mentioned in \cref{sec:introduction}, it is hard to directly use \Cref{thm:canonical-tester} to design and analyze our streaming testers due to the intersections of $q$-RBFS. In order to tackle this difficulty, we decompose each forbidden subgraph $F \in \mathcal{F}_n$ into all possible sets of intersecting $q$-bounded discs whose union is $F$. In order to recover $F$ from such a decomposition, we have to identify and monitor \emph{vertices that are contained in more than one $q$-bounded disc of $F$}.

\subsubsection{Identifying vertices with large reach probability}

In this section, we prove that with constant probability the $q$-bounded discs found by $q$-RBFS will only intersect on a \emph{small} set of vertices $V_\alpha$ and the discs will not intersect on any edge.

We begin with a useful definition on the probability of reaching a vertex from a $q$-RBFS.

\begin{definition}
\label{def:reach-probability}
For each vertex $v$, the \emph{reach probability $r(v) \defeq r_q(v)$} of $v$ is the probability that a $q$-RBFS starting at a uniformly randomly chosen vertex reaches $v$.
\end{definition}

In the following lemma, we give an upper bound on the size of the set of vertices with constant reach probability, which also implies that with constant probability, the number of vertices visited by at least two $q$-RBFS that the canonical tester performs is small. For any $\alpha$, $0 \le \alpha \le 1$, we let $V_\alpha := \{v \in V : r(v) \ge \alpha\}$. For a fixed $q$, let $c_j := \sum_{i=0}^{j} q^i = \frac{q^{j+1}-1}{q-1}$.

\begin{lemma}
\label{thm:number-collisions}
For any $0 < \alpha < 1$, it holds that $|V_\alpha| \le \frac{c_q}{\alpha}$.
\end{lemma}

\begin{proof}
Let $p_u(v)$ be the probability that a $q$-RBFS starting at vertex $u$ discovers vertex $v$. Note that $r(v) = \frac{1}{n} \sum_u p_u(v)$.

Let $X_u$ denote the number of vertices in the subgraph explored by a $q$-RBFS starting at vertex $u$. Let $X$ denote the number of vertices in the subgraph explored by a $q$-RBFS starting at a vertex that is chosen i.u.r. from $V$. Note that $\E[X_u] = \sum_v p_u(v)$, and thus,
\begin{align*}
    \E[X]
        &=
    \frac{1}{n} \sum_u \Ex[X_u] =
    \sum_u \frac{1}{n} \sum_v p_u(v) =
    \sum_v\frac{1}{n} \sum_u p_u(v) =
    \sum_v r(v)
        \fstop
\end{align*}

We observe that a $q$-RBFS starting at an arbitrary vertex $u$ explores at most $q^i$ vertices at distance $i$ from $u$, which gives that $\Ex[X_u] \le \sum_{i=0}^{q} q^i = c_q$, and hence also $\Ex[X] \le c_q$.

Recall that $V_\alpha = \{v: r(v) \ge \alpha\}$. We have that		
\begin{equation*}
	c_q \ge \Ex[X] =
    \sum_v r(v) \ge
    \sum_{v \in V_\alpha} r(v) \ge
    |V_\alpha| \cdot \alpha
        \enspace,
\end{equation*}
which concludes the lemma.
\end{proof}

We further show that with high probability, two $q$-RBFS starting from vertices chosen i.u.r.\ will not share an edge (i.e., will not visit the same edge).

\begin{lemma}
\label{thm:random-bfs-visit-edges}
Let $0 < \alpha \le 1$. Let $n \ge \frac{qc_q}{\alpha^2}$. Let $u,v$ be two randomly chosen vertices. Let $H_u$ and $H_v$ denote the subgraphs visited by two $q$-RBFS starting at $u$ and $v$, respectively. Then with probability $1 - q c_q \cdot  2 \alpha$, no edge will be contained in both $H_u$ and $H_v$.
\end{lemma}

In order to prove the above lemma, we first show that vertices with a large reach probability have large degree (that is linear in $n$). %

\begin{lemma}
\label{thm:collision-high-degree}
Let $0 < \alpha \le 1$. It holds that for any $v\in V_\alpha$, $\deg(v)\ge \frac{n\alpha}{c_q}.$
\end{lemma}

\begin{proof}
Let $H=(S,E(H))$ be the subgraph explored by a $q$-RBFS starting at a vertex that is chosen uniformly at random from $V$. For each $0 \le i \le q$, we let $S_i \subseteq S$ denote the set of vertices at distance exactly $i$ from the root of $S$.
	
For any $v \in V$, let $p(v)$ be the probability that $v$ is contained in $S$. We have the following claim (the proof is given subsequent to this one).
	
\begin{claim}
\label{ref:claim-who-knows-what-for}
Let $0\le i\le q$. For every non-isolated 	vertex $v \in V$, conditioned on the event that $v$ is not contained in $\cup_{j\le i-1} S_j$, the probability that $v$ is contained in $S_i$ is at most $\frac{q^i \deg(v)}{n}$.
\end{claim}
	
By \cref{ref:claim-who-knows-what-for}, the probability that $v$ is contained in $S$ is at most
\begin{align*}
	\sum_{i=0}^q \frac{q^i \degr{v}}{n} \le \frac{c_q \degr{v}}{n}
	\fstop
\end{align*}
	
Recall that $V_\alpha = \{v: r(v) \ge \alpha\}$. By noting that $r(v)$ is exactly the probability that $v$ is contained in $S$, we have that for any $v\in V_\alpha$,
\begin{align*}
	\alpha &\le \frac{c_q \degr{v}}{n}
	\enspace,
\end{align*}
which gives that $\deg(v) \ge \frac{n\alpha}{c_q}$. This completes the proof of \cref{thm:collision-high-degree}.
\end{proof}

\begin{proof}[Proof of \cref{ref:claim-who-knows-what-for}]
We prove the above claim by induction on $i$. If $i = 0$, then the probability that the $q$-RBFS visits $v$ is $\frac{1}{n}$, and thus $v$ is contained in $S_0$ with probability at most $\frac{1}{n}$.
	
Let us assume now that the statement of the claim holds for $i-1$. That is, for any vertex $u$, the probability that $u$ is contained in $S_{i-1}$ is at most $\frac{q^{i-1}\deg(u)}{n}$.
	
Consider an arbitrary vertex $v$. By induction, for every neighbor $u$ of $v$, $u$ is contained in $S_{i-1}$ with probability at most $\frac{q^{i-1} \degr{u}}{n}$. In the $q$-RBFS, each vertex $w$ in $S_{i-1}$ samples $q$ neighbors of $w$ i.u.r., which implies that the probability that $v$ is contained in $S_i$ is at most
\begin{align*}
	\sum_{u \in \ngh{v}} \frac{q^{i-1} \degr{u}}{n} \sum_{j=1}^q \frac{1}{\degr{u}}
	   &=
	\sum_{u \in \ngh{v}} \frac{q^i}{n}
	   =
	\frac{q^{i} \degr{v}}{n}
\enspace,
\end{align*}
where $\ngh{v}$ is the set of neighbors of $v$ in $G$. This yields the proof of \cref{ref:claim-who-knows-what-for}.
\end{proof}

Now we are ready to prove \cref{thm:random-bfs-visit-edges} which upper bounds the probability that the two subgraphs explored by two $q$-RBFS starting at two random vertices share any edge.

\begin{proof}[Proof of \cref{thm:random-bfs-visit-edges}]
Let us consider an arbitrary edge $(x,y)\in E(H_u)$.
\begin{enumerate}
\item $x,y \in V_\alpha$: By \cref{thm:collision-high-degree}, if $v\in V_\alpha$, then $\deg(v) \ge \frac{n\alpha}{c_q}$. Suppose that at least one of $x,y$, say $x$, is also discovered by the $q$-RBFS from $v$, i.e., $x\in V(H_v)$ (otherwise, $(x,y)$ will not be contained in $H_v$ at all). Thus, the probability that $(x,y)$ will be contained in $H_v$ (i.e., visited by $q$-RBFS from $v$) is at most $ \frac{q}{\deg(u)} \le \frac{q c_q}{n \alpha} \le \alpha$. By the union bound, $(x,y)$ will be contained in $H_v$ with probability at most $\frac{2 q c_q}{n \alpha}$.
\item $x,y \notin V_\alpha$: By definition of $V_\alpha$, the probability that $x$ is contained in $V(H_v)$ is at most $\alpha$ (and similarly for $y$). By the union bound, $(x,y) \in E(H_v)$ with probability at most $2 \alpha$.
\item $x \in V_\alpha, y \notin V_\alpha$:  By the above analysis, if $x$ is contained in $V(H_v)$, then $(x,y)$ will be contained in $E(H_v)$ with probability at most $\frac{q c_q}{n \alpha}$. Further note that $y$ will be contained in $V(H_v)$ with probability at most $\alpha$. Thus, the probability that $(x,y) \in E(H_v)$ is at most $\frac{q c_q}{n \alpha} + \alpha \le 2 \alpha$.
\item $x \notin V_\alpha, y \in V_\alpha$: By similar analysis to the above item, the probability that $(x,y) \in E(H_v)$ with probability at most $\frac{q c_q}{n \alpha} + \alpha \le 2 \alpha$.
\end{enumerate}
Furthermore, we note that $|E(H_u)|\le q c_{q}$. This implies that with probability at least $1 - q c_q \cdot 2 \alpha$, none of edges in $E_\alpha$ will be contained in $H_v$.
\end{proof}

\subsubsection{Colored $q$-bounded disc types}

To identify vertices in $V_\alpha$, we assign them unique colors for the analysis. We call a disc \emph{$r$-colored} if in addition to uncolored vertices in the disc, some vertices in the disc may be colored with at most $r$ colors, each color being used at most once. Two colored $q$-bounded disc types $\Delta_1$ and $\Delta_2$ (cf. \cref{def:bounded-disc-type})
are called to be isomorphic to each other, denoted by $\Delta_1\simeq \Delta_2$, if there is a root-preserving isomorphism $f$ from $\Delta_1$ to $\Delta_2$ that also preserves the colors, \ie if and only if $u\in V(\Delta_1)$ is colored with color $c$, then $f(u) \in \Delta_2$ is colored with color $c$.

\begin{definition}
Let $q>0$ be an integer. We let $\mathcal{H}_{q} := \{\Delta_1, \cdots, \Delta_N\}$ denote the set of all possible $r$-colored $q$-bounded disc types, where $N$ is the total number of such types.
\end{definition}

For any given colored $q$-bounded disc type, we have the following definition on the probability of seeing such a disc type from a $q$-RBFS.

\begin{definition}[\textbf{Reach probability of colored $q$-bounded disc types}]
Let $G=(V,E)$ be a graph with $n$ vertices such that each vertex in $V_\alpha$ is assigned to a unique color. Let $\Delta\in \mathcal{H}_q$ be a colored $q$-bounded disc type. The \emph{reach probability of $\Delta$ in $G$} is the probability that a $q$-RBFS from a random vertex in $G$ reveals a graph that is (root- and color-preserving) isomorphic \footnote{Given two rooted graphs $G,H$, a root-preserving isomorphism from $G$ to $H$ is a bijection $f:V(G)\rightarrow V(H)$ such that if $u$ is the root of $V(G)$, then $f(u)$ is the root of $V(H)$; that if $u$ is a colored vertex, then $f(u)$ is also a colored vertex; that $(u,v)\in E(G)$ if and only if $(f(u),f(v))\in E(H)$; and that two colored vertices $u,v$ have different colors if and only if $f(u),f(v)$ have different colors.}
to $\Delta$, that is
\begin{align*}
	\reach_G(\Delta) &:= \Pr_{v\sim V,BFS}[\textsc{RandomBFS}(G,v,q) \simeq \Delta]
    \fstop
\end{align*}

For a given vertex $v$, the reach probability of $\Delta$ from $v$ in $G$ is the probability that a $q$-RBFS from $v$ in $G$ induces a graph that is (root- and color-preserving) isomorphic to $\Delta$, that is
\begin{align*}
	\reach_G(v, \Delta) &:= \Pr_{BFS}[\textsc{RandomBFS}(G,v,q) \simeq \Delta]
    \fstop
\end{align*}
\end{definition}

Recall from \cref{def:bounded-disc} that a $q$-bounded disc of $v$ in $G$ is any subgraph $H$ of $G$ that is rooted at $v$ and can be returned by \textsc{RandomBFS}$(G,v,q)$. In order to estimate the reach probability of a colored $q$-bounded disc type, we consider for each starting vertex $v$, the set of all possible colored $q$-bounded discs, called colored $q$-bounded discs, that one can see from a $q$-RBFS from $v$.

\begin{definition}[\textbf{Reach probability of a $q$-bounded disc}] \label{def:configuration}
Let $G=(V,E)$ be a graph in which all vertices in $V_\alpha$ are uniquely colored. Let $v$ be a vertex in $G$. A \emph{colored $q$-bounded disc} of $v$ is a $q$-bounded disc of $v$ in $G$ with all vertices in $V_\alpha$ colored. We let $\mathcal{C}_v$ denote the set of all possible colored $q$-bounded discs of $v$.\footnote{Note that the number $|\mathcal{C}_v|$ of colored $q$-bounded discs of $v$ can be arbitrarily large.}

For any fixed colored $q$-bounded disc $C\in\mathcal{C}_v$ of $v$, the reach probability of $C$ from $v$ is the probability that a $q$-RBFS from $v$ sees exactly $C$, that is,
\begin{align*}
	\reach_G(v, C) &:= \Pr_{BFS}[\textsc{RandomBFS}(G,v,q) = C]
    \fstop
\end{align*}
\end{definition}

By our definition, the $q$-RBFS from a vertex $v$ in the colored graph $G$ (with vertices in $V_\alpha$ colored) will return exactly one colored $q$-bounded disc of $v$. For each colored $q$-bounded disc type $\Delta$, we let $\mathcal{C}_v(\Delta)$ denote the subset of $\mathcal{C}_v$ which contains all colored $q$-bounded discs of $v$ that are isomorphic to $\Delta$. Therefore, we have the following observation.
\begin{align}
\label{eqn:reach-delta-config}
	\reach_G(v,\Delta) &= \sum_{D\in \mathcal{C}_v(\Delta)}\reach_G(v,D)
    \fstop
\end{align}

\subsubsection{Canonical testers with distinguished vertices in the intersecting discs}
\label{subsubsec:canonical-tester-refined}

Now, we give a refined characterization of the family of forbidden subgraphs corresponding to any constant-query testable property in general graphs, which establishes the basis of our framework for transforming the canonical constant-query testers in the random neighbor model to the random order streaming model.

In our next theorem, we will consider partially vertex-colored graphs and $q$-bounded discs: we color each vertex in $V_\alpha$ with a unique color from a palette of size $|V_\alpha|$. Recall from \cref{thm:number-collisions} that $|V_\alpha|\le \frac{c_q}{\alpha}$. %
We obtain canonical testers of constant-query testable properties by \emph{forbidden colored $q$-bounded discs} instead of \emph{forbidden subgraphs} (that can be composed of more than a single $q$-bounded disc). See \cref{fig:stitching} for an example.

\begin{theorem}\label{thm:refined_characterization}
Let $\Pi = (\Pi_n)_{n \in \setn}$ be a graph property with query complexity $q=q(\varepsilon)$ and let $\varepsilon > 0$, $\alpha \le \frac{1}{12 (q')^2}$, where $q'$ is the number from \cref{thm:canonical-tester}. There is an infinite sequence $\mathcal{F}' = (\mathcal{F}'_n)_{n \in \setn}$ such that for any $n \ge \frac{qc_q}{\alpha^2}$, the following properties hold:
\begin{itemize}
\item $\mathcal{F}'_n$ is a set of graphs, and for each graph $F\in \mathcal{F}'_n$, there exists at least one multiset $S$ of $q'$ many $c_q/\alpha$-colored and rooted $q'$-bounded disc types such that 1) the disc types are pairwise edge-disjoint, and 2) the graph obtained by identifying all vertices of the same color in the bounded discs of $S$ is isomorphic to $F$.
\item For any $n$-vertex graph $G=(V,E)$ such that each vertex in $V_\alpha$ is colored uniquely, let $S_{q'}$ denote the set of $q'$ subgraphs obtained by performing $q'$-RBFS starting at $q'$ vertices sampled i.u.r. Then,
\begin{itemize}
\item if $G \in \Pi_n$, with probability at least $\frac23$, there is no $F \in \mathcal{F}'_n$ such that $F \simeq S_{q'}$,
\item if $G$ is $\varepsilon$-far from $\Pi_n$, with probability at least $\frac23$, there exists $F \in \mathcal{F}'_n$ such that $F \simeq S_{q'}$,
\end{itemize}
where the probability is taken over the randomness of $S_{q'}$.
\end{itemize}
Furthermore, if $\Pi$ can be tested with one-sided error, then for $G\in \Pi_n$, with probability $1$, there is no $F \in \mathcal{F}'_n$ such that $F \simeq S_{q'}$.
\end{theorem}

\begin{proof}
Let $\mathcal{T}$ be a canonical tester for $\Pi$ with error-probability $\frac16$ that is obtained by applying \cref{thm:canonical-tester}, and let $\mathcal{F}_n$ be the corresponding family of forbidden graphs. We prove that we can decompose every $F \in \mathcal{F}_n$ into a family of multisets of colored $q'$-bounded discs, and $\mathcal{F}'_n$ can be constructed as the set of all these families.
	
Let $F \in \mathcal{F}_n$. Let $\mathcal{D}$ be the maximal family of multisets of $q'$ many $q'$-bounded disc types such that for every $D \in \mathcal{D}$, there exists a mapping $f$ from the vertices of the $q'$-bounded disc types in $D$ to the vertices of $F$ with the following properties: (i) $f$ is surjective, (ii) $f$ restricted to a single $\Delta \in D$ is injective, (iii) $f$ restricted to root vertices in $F$ is a bijection, (iv) for every $\Delta \in D$ and $u,v \in V(\Delta)$, $(u,v) \in E(D)$ if and only if $(f(u),f(v)) \in E(F)$.
	
Let $F$ be the subgraph observed by $\mathcal{T}$ and let $D \in \mathcal{D}$ be a corresponding multiset of $q'$ many $q'$-bounded disc types. By \cref{thm:number-collisions}, there exist at most $c_q / \alpha$ vertices (i.e., vertices in $V_\alpha$) in $F$ such that there exist $\Delta_1, \Delta_2 \in D$ and $v_1 \in V(\Delta_1)$, $v_2 \in V(\Delta_2)$ such that $f(v_1) = f(v_2)$. Furthermore, by \cref{thm:random-bfs-visit-edges} and the union bound, with probability at least $1 - (q')^2 \cdot 2 \alpha$, there is no pair $\Delta_1, \Delta_2 \in D$ such that there exists $(u_1, v_1) \in E(\Delta_1), (u_2, v_2) \in E(\Delta_2)$ and $(f(u_1), f(v_1)) = (f(u_2), f(v_2))$. Therefore, we can color all vertices in $V_\alpha$ with at most $c_q / \alpha$ colors and decompose $D$ into a multiset of $q'$ colored $q'$-bounded discs types such that there exists a bijection between the roots of the $q'$-bounded discs and the rooted vertices in $F$. Note that this decomposition is not necessarily unique. See \cref{fig:stitching} for an example.
	
Finally, let $H_1, \ldots, H_{q'}$ be the $q$-bounded discs that are found by a run of the canonical tester. The tester errs with probability $\frac{1}{6}$, which implies that $H_1 \cup \ldots \cup H_{q'}$ is not contained in $\mathcal{F}_n$. Assume that this is not the case, i.e., $H_1 \cup \ldots \cup H_{q'}$ is contained in $\mathcal{F}_n$. Let $\Delta_1, \ldots, \Delta_{q'}$ be the $q$-bounded disc types corresponding to $H_1, \ldots, H_{q'}$. Then, with probability at most $2(q')^2 \alpha$, the graph obtained by identifying vertices from $\Delta_1, \ldots, \Delta_{q'}$ that have the same color is not contained in $\mathcal{F}'_n$. The claim then follows from the fact that the total error probability is $\frac{1}{6} + 2(q')^2 \alpha \le \frac{1}{3}$.
\end{proof}

\section{Estimating the Reach Probabilities in Random Order Streams}
\label{sec4}

Given a canonical tester $\mathcal{T}$ for a property $\Pi$ that is constant-query testable in the random neighbor model, we transform it into a random-order streaming algorithm as follows. Recall from \cref{thm:canonical-tester} that $\mathcal{T}$ explores the input graph by sampling vertices uniformly at random and running $q$-RBFS for each of these vertices. Only if the resulting subgraph contains an instance of a forbidden subgraph from a family $\mathcal{F}$, it rejects. It seems natural to define a procedure like $q$-RBFS for random order streams, namely a procedure $\textsc{StreamCollect}(\stream(G), v, q)$ (\emph{$q$-SC}), and let the streaming algorithm reject only if the union of all $q$-SC contains an instance of a graph from $\mathcal{F}$. However, this raises a couple of issues.

It seems hard to analyze the union of the subgraphs obtained by $q$-SC and relate it to the union of subgraphs observed by $q$-RBFS because the interference between two $q$-SC is quite different from the interference of two $q$-RBFS. Therefore, we use \cref{thm:refined_characterization}, which roughly says that we can decompose each forbidden subgraph into colored $q$-bounded disc types. This leads to the following idea: First, we prove that for any colored $q$-bounded disc type $\Delta$, if $q$-RBFS finds an instance of $\Delta$ in the input graph with probability $p$ (where colors correspond to intersections of multiple RBFS), then $q$-SC finds an instance of $\Delta$ with probability $cp$ for some suitable constant $c$. Then, we prove that if $S$ is a sufficiently large set of vertices sampled uniformly at random, for each colored $q$-bounded disc type $\Delta$, the fraction of $q$-bounded discs found by $q$-SCs started from $S$ that are isomorphic to $\Delta$ is bounded from below by the probability that a $q$-RBFS from a random vertex sees a colored $q$-bounded disc that is isomorphic to $\Delta$. Finally, in the next section, we conclude that if $q$-RBFS finds a forbidden subgraph $F \in \mathcal{F}$ with probability $p$, then the fraction of $q$-SC also finds this subgraph with probability $cp$ (for some suitable constant $c$) because it will find the corresponding colored $q$-bounded discs that assemble $F$.

\subsection{Collecting a $q$-Bounded Disc in a Graph Stream}

In our streaming algorithm, we need to collect a $q$-bounded disc from a starting vertex $v$. We do this in a natural and greedy way: We start with a graph $H=(U,F)$ with $U=\{v\}$ and $F=\emptyset$. Then whenever we see an edge $(u,w)$ from the stream that is connected to our current graph $H$ and adding $(u,w)$ to $H$ does not violate the $q$-bounded radius of $H$, and the degree of $u$ or the degree of $w$ in $H$ is still less than $q^{2q}$, we add it to $F$ (and possibly add one of its endpoint to $U$); otherwise, we simply ignore the edge. Note that the algorithm does not assign colors to the subgraphs it explores. The procedure is formally defined in \cref{alg:collect-nbh-stream}.

\begin{algorithm}
	\caption{Collecting a $q$-bounded disc from a vertex in stream}
	\label{alg:collect-nbh-stream}
	\begin{algorithmic}
		\Function{StreamCollect}{$\stream(G), v, q${}}	
			\State $U\gets \{v\}$
			\State $\forall u \in V: d_u \gets \infty, \ell_u \gets 0$
			\State $d_v \gets 0; F \gets \emptyset$
			\State $H=(U,F)$ with $v$ marked as root
			\For{$(u,w)\gets$ next edge in the stream}
				\If{$(\{u, w\} \cap U \neq \emptyset)$}
					\If{$(u \in U \Rightarrow (\ell_u < q \wedge d_u < q^{2q}) \vee (w \in U \Rightarrow (\ell_w < q \wedge d_w < q^{2q}))$}
						\State $U \gets U \cup \{ u, w \}$
						\State $F \gets F \cup ( u, w )$
						\State $d_u \gets d_u + 1; d_w \gets d_w + 1$
						\State $\ell_u \gets \min(\ell_u, \ell_w + 1); \ell_w \gets \min(\ell_w, \ell_u+1)$
					\EndIf
				\EndIf
			\EndFor
			\State \Return $H$
		\EndFunction
	\end{algorithmic}
\end{algorithm}

\subsection{Relation of One $q$-SC and One $q$-RBFS}

In the following, we show that for any vertex $v$, and any colored $q$-bounded disc $C$ of $v$, the probability of collecting $C$ from $v$ by running \textsc{StreamCollect} on a random order edge stream is at least a constant factor of the probability of reaching $C$ from $v$ by running a $q$-RBFS on $G$. The statements in this section hold for a single run of $q$-SC.

We emphasize that the coloring does not need to be explicitly given. It is sufficient if it can be applied when random access to the graph is given. In particular, we may assign each vertex in $V_\alpha$ a unique color. This enables us to identify the vertices where multiple $q$-RBFS may intersect, which is crucial to apply \cref{thm:refined_characterization} later.

\begin{lemma}
\label{lemma:boundconfiguration}
Let $G$ be a vertex-colored graph. There exists a constant $\cst(q)$ depending on $q$, such that for any colored $q$-bounded disc $C$, it holds that
\begin{align*}
    \Pr_{\stream(G)}[\textsc{StreamCollect}(\stream(G),v,q) \text{ contains $C$}] &\ge \cst(q)\cdot \reach_G(v,C)
    \fstop
\end{align*}
\end{lemma}

\begin{proof}
Note that both \textsc{RandomBFS} and \textsc{StreamCollect} can be viewed as the random processes of revealing vertices and edges in $G$. Such a process starts from the fixed vertex $v_0:=v$. At each time step $t\ge 0$, some new edges and vertices (which are added to the queue $S$) are revealed.

For any fixed $v$, and any of its colored $q$-bounded disc $C$, we call an edge ordering $\sigma$ over $E(C)$ \emph{good} if it can be \emph{realized} in a $q$-RBFS from $v$ that discovers $C$. Note that the ordering $\sigma$ also defines an ordering over the vertices, which corresponds to the ordering of popping the vertex from the queue in \textsc{RandomBFS}($G,v,q$).

We note that
\begin{align*}
    \reach_G(v,C)
        =
    \sum_{\sigma: \text{good edge ordering}} \Pr_{RBFS}[\sigma]
    \fstop
\end{align*}

Now let us consider an arbitrary good edge ordering $\sigma$ over $E(C)$. Let $v_0:=v,v_1,\cdots,v_k$ be the corresponding vertex ordering $v_0:=v,v_1,\cdots,v_k$ over $V(C)$, where $k=|V(C)|\le q^{2q}$.

We let $\ent_t$ denote the event that for any $0\le i\le t$, when $v_i$ is popped out from the queue, the edges $(v_i,w_{i,1}), \cdots, (v_i,w_{i,j_i})$ are sampled out in the same order as as the one defined by $\sigma$, where $j_i\le q$. Then it holds that
\begin{align*}
    \Pr_{RBFS}[\sigma]
        &=
    \Pr_{RBFS}[\cap_{i\le k} \ent_i]
        =
    \Pr_{RBFS}[\ent_0] \cdot \Pr_{RBFS}[\ent_1|\ent_0]
        \cdots \Pr_{RBFS}[\ent_k|\cap_{i\le k-1} \ent_i]
    \fstop
\end{align*}
Now we note that by the definition of $q$-RBFS, it holds that for any $t\le k$,
\begin{align*}
    \Pr_{RBFS}[\ent_t| \cap_{i\le t-1} \ent_i]
        &= \alpha(j_t)\left(\frac{1}{\deg(v_t)}\right)^q
    \enspace,
\end{align*}
where $\alpha(j_t)$ is a constant depending on $j_t\le q$.

Now we consider the probability of seeing this edge ordering in random streaming order. Then it holds that
\begin{align*}
    \Pr_{\stream(G)}[\sigma]
        &=
    \Pr_{\stream(G)}[\cap_{i\le k} \ent_i]
        =
    \Pr_{\stream(G)}[\ent_0] \cdot \Pr_{\stream(G)}[\ent_1|\ent_0]
        \cdots \Pr_{\stream(G)}[\ent_k|\cap_{i\le k-1} \ent_i]
    \fstop
\end{align*}

Now recall that we let $j_i$ denote the number of edges that are collected from $v_i$ in the ordering $\sigma$. Let $s_i:=\sum_{r\le i} j_i$ denote the number of edges after collecting edges from $v_i$. By the definition of random ordering of the stream, conditioned on $\cap_{i\le t-1} \ent_i$, the probability of seeing the next $j_t$ edges from $v_t$ is at least the probability that all the next $j_t$ edges appear after the first $s_{t-1}$ edges, times the probability that these $j_{t}$ edges appear earlier than the remaining edges incident to $v_t$. That is, for any $t\le k$,
\begin{align*}
    \Pr_{\stream(G)}[\ent_t| \cap_{i\le t-1} \ent_i]
        &\ge
    \left(\frac{1}{s_{t-1}+1}\right)^{j_t} \!\!\!
        \min_{\lambda_t: 0 \le \lambda_t \le \deg(v_t)-j_t}
            \frac{(\deg(v_t)-\lambda_t-j_t)!}{(\deg(v_t)-\lambda_t)!}
        \ge
    \beta(j_t,s_{t-1}) \cdot \left(\frac{1}{\deg(v_t)}\right)^q
    \enspace,
\end{align*}
where $\lambda_t$ denotes the possible number of edges incident to $v_t$ appeared before the time we collect edges from $v_t$, $\beta(j_t,s_{t-1})$ is a constant depending on $j_t\le q$ and $s_{t-1}\le tq$.

Finally, we note that
\begin{align*}
    \lefteqn{\Pr_{\stream(G)}[\textsc{StreamCollect}(\stream(G),v,q) \text{ contains $C$}]}
        &
        \\
        &\ge
    \sum_{\sigma: \text{good edge ordering}} \Pr_{BFS}[\sigma]
        \\
        &\ge
    \sum_{\sigma: \text{good edge ordering}} \Pr_{\stream(G)}[\ent_0^\sigma] \cdot \Pr_{\stream(G)}[\ent_1^\sigma|\ent_0^\sigma] \cdots \Pr_{\stream(G)}[\ent_k^\sigma|\cap_{i\le k-1} \ent_i^\sigma]
        \\
        &\ge
    \sum_{\sigma: \text{good edge ordering}} \prod_{t=0}^{k} \frac{\beta(j_t,s_{t-1})}{\alpha(j_t)} \Pr_{RBFS}[\ent_0] \cdot \Pr_{RBFS}[\ent_1^\sigma|\ent_0^\sigma] \cdots \Pr_{RBFS}[\ent_k^\sigma|\cap_{i\le k-1} \ent_i^\sigma]
        \\
        &=
    \cst(q) \sum_{\sigma: \text{good edge ordering}} \Pr_{RBFS}[\sigma]
        \\
        &=
    \cst(q) \cdot \reach_{G}(v,C)
    \enspace,
\end{align*}
where we defined $\cst(q):=\prod_{t=0}^{k} \frac{\beta(j_t,s_{t-1})} {\alpha(j_t)}$. This finishes the proof of \cref{lemma:boundconfiguration}.
\end{proof}

The following lemma performs the step from $q$-bounded discs to $q$-bounded disc \emph{types}.

\begin{lemma}
Let $\Delta$ be a fixed colored $q$-bounded disc type. Let $X_v$ denote the indicator variable that \textsc{StreamCollect} from $v$ collects a subgraph that contains a colored $q$-bounded disc of $v$ that is isomorphic to $\Delta$. Let $Y_v$ denote the indicator variable that \textsc{RandomBFS} from $v$ sees a colored $q$-bounded disc of $v$ that is isomorphic to $\Delta$. Then it holds that
\begin{align*}
    \E_{\stream(G)}[X_v] \ge \cst(q)\cdot \E_{RBFS}[Y_v]
    \enspace,
\end{align*}
where $\cst(q)$ is the constant from \cref{lemma:boundconfiguration}.
\end{lemma}

\begin{proof}
Let $C\in \mathcal{C}_v(\Delta)$ denote any colored $q$-bounded disc that is isomorphic to $\Delta$. Let $X_{v,C}$ be the indicator variable that \textsc{StreamCollect} from $v$ collects a subgraph that contains $C$. Let $Y_{v,C}$ be the indicator variable that \textsc{RandomBFS} from $v$ sees $C$. Then by linearity of expectation, we have
\begin{align*}
    \E_{\stream(G)}[X_v]
        &=
    \sum_{C\in \mathcal{C}_v(\Delta)}\E_{\stream(G)}[X_{v,C}], \quad \E_{RBFS}[Y_v]=\sum_{C\in \mathcal{C}_v(\Delta)}\E_{RBFS}[Y_{v,C}]
    \fstop
\end{align*}
Note that $\E_{\stream(G)}[X_{v,C}] = \Pr_{\stream(G)}[\textsc{StreamCollect}(\stream(G),v,q) \text{ contains $C$}]$ and that $\E_{RBFS}[Y_{v,C}]=\reach_G(v,C)$. Then the statement of the lemma follows from \cref{lemma:boundconfiguration}.
\end{proof}

Now we consider the probability of seeing a colored $q$-disc type $\Delta$. Note that $\E_{\stream(G)}[X_v] = \Pr_{\stream(G)}[\textsc{StreamCollect}(\stream(G),v,q)\text{ contains a subgraph $F$ with $F \simeq \Delta$}]$. Furthermore, $\E_{RBFS}[Y_v]=\reach_G(v,\Delta)$. Thus, we have the following lemma.

\begin{corollary}
\label{lemma:streamproblower}
For any colored $q$-bounded disc type $\Delta$, it holds that
\begin{align*}
    \Pr_{\stream(G)}[\textsc{StreamCollect}(\stream(G),v,q)\text{ contains a subgraph $F$ with $F \simeq \Delta$}]
        &\ge \cst(q)\cdot \reach_G(v,\Delta)
    \fstop
\end{align*}
\end{corollary}

\subsection{Relation of Multiple $q$-SCs and $q$-RBFS}

In the previous section, we related a single run of $q$-RBFS and a single run of $q$-SC. In particular, \cref{lemma:streamproblower} states that if a $q$-RBFS starting from $v$ finds some colored $q$-bounded disc type $\Delta$ with probability $p$, $q$-SC finds the same type $\Delta$ with probability $\Omega(p)$. However, the forbidden subgraphs that the property tester aims to find may be composed of more than one $q$-bounded disc. Therefore, we need to prove that if multiple runs of $q$-RBFS find $q$-bounded disc types $\Delta_1, \ldots, \Delta_k$ whose union contains an instance of a forbidden subgraph $F \in \mathcal{F}'_n$, then multiple runs of $q$-SC will find $\Delta_1, \ldots, \Delta_k$ with probability $\Omega(p)$.

We now show our main technical lemma on estimating the reach probability of $q$-bounded disc types in random order streams. Again, the coloring of vertices in $G$ is implicit and only used for the analysis.

\begin{lemma}
\label{lemma:lower_bound_disc_probability}
Let $G=(V,E)$ be a graph with all vertices in $V_\alpha$ colored defined by a random order stream and let $q>0$ be an integer. Let $c'_q:=\sum_{i=0}^{q+1}q^{2qi}$. Let $\delta>0$, $\alpha = \frac{\delta^6}{6400|\mathcal{H}_q|^2q^{2q}c'_q}$ and let $S$ denote a set of vertices that are chosen uniformly at random with $|S|=s\ge \max\{\frac{1}{20\sqrt{\alpha q^{2q}\cdot c'_q}},\frac{5000|\mathcal{H}_q|}{\cst(q) \delta^3} \} $. Let $\mathcal{J}:=\{H_v: H_v=\textsc{StreamCollect}(\stream(G),v,q), v\in S\}$ denote the set of colored $q$-bounded discs collected by \textsc{StreamCollect} from vertices in $S$. For each type $\Delta\in \mathcal{H}_q$, let $X_\Delta$ denote the number of graphs $H$ in $\mathcal{J}$ such that $H$ contains a subgraph $F$ with $F\simeq \Delta$.

Then it holds that with probability at least $1 - \frac{1}{100}$, for each type $\Delta\in \mathcal{H}_q$,
\begin{align*}
    q_\Delta :=
    \frac{1}{\cst(q)} \cdot \frac{X_\Delta}{s} \ge \reach_G(\Delta) - \delta
    \enspace,
\end{align*}
where $\cst(q)$ is a constant from \cref{lemma:streamproblower}.
\end{lemma}

\begin{proof}
We first note that we only need to consider $\Delta$ with $\reach_G(\Delta)\ge \delta$. As otherwise, the statement of the lemma trivially holds. Since we sampled a set $S$ with $|S|\ge \Omega(\frac{\log(|\mathcal{H}_q|)}{\delta^2})$, the following claim follows from the Chernoff bound.

\begin{claim}\label{claim:sample_reach_prob}
With probability (over the randomness of sampling $S$) at least $1-\frac{1}{400|\mathcal{H}_q|}$, it holds that
\begin{align}
\label{ineq:delta-average}
\left|\frac{\sum_{v\in S}\reach_{G}(v,\Delta)}{|S|}-\reach_G(\Delta)\right|
&\le
\frac{\delta}{2}
\fstop
\end{align}
\end{claim}
Furthermore, similar to the proof of \cref{thm:random-bfs-visit-edges}, we have the following claim (the proof is deferred to the end of this section).

\begin{claim}
\label{claim:streamcollectvisitedges}
Let $\alpha$ be $0 < \alpha \le 1$. Let $c'_q:=\sum_{i=0}^{q+1}q^{2qi}$. Let $\alpha_0:=\alpha q^{2q}\cdot c'_q$. Let $n \ge \frac{q^{2q}c'_q}{\alpha^2}$.
Let $H_u$ and $H_v$ denote the subgraphs collected by the \textsc{StreamCollect} starting at two randomly chosen vertices $u$ and $v$, respectively. Let
\begin{align*}
    Y_{uv} &:= \Pr_{\stream(G)}[\text{$H_u$ and $H_v$ share some edge}]
    \fstop
\end{align*}
Then with probability (over the randomness of choosing $u,v$) at least $1-2\sqrt{\alpha_0}$, it holds that
\begin{equation}
\label{eqn:stream_sharing_edge}
	Y_{uv} \le \sqrt{\alpha_0}
    \fstop
\end{equation}
\end{claim}

In the following, we will condition on the following event $\mathcal{E}$: \cref{ineq:delta-average} holds as stated in \cref{claim:sample_reach_prob}, and  \cref{eqn:stream_sharing_edge} holds as stated in \cref{claim:streamcollectvisitedges}. By \cref{claim:sample_reach_prob} and \cref{claim:streamcollectvisitedges}, the event $\mathcal{E}$ holds with probability (over the randomness of sampling vertex set $S$) at least $1-\frac{1}{400 |\mathcal{H}_q| }\cdot |\mathcal{H}_q| - \alpha q^{2q}\cdot c'_q \cdot s^2\ge 1-\frac{1}{200}$ by our choice of $s$ and $\alpha$.

Let us now consider a fixed type $\Delta\in \mathcal{H}_q$.
By \cref{lemma:streamproblower}, for any fixed $v$, we know that $\E_{\stream}[X_v]\ge \cst(q)\cdot \reach_G(v,\Delta)$. Therefore, our estimate $q_\Delta$ for $\Delta$ satisfies that
\begin{align*}
    \E_{\stream}[q_\Delta]
        &\ge
    \frac{1}{\cst(q)}\cdot \cst(q)\cdot \frac{\sum_{v\in S}\reach_G(v,\Delta)}{|S|}
        \ge
    \reach_G(\Delta)-\frac{\delta}{2}
\enspace,
\end{align*}
where the last inequality follows from \cref{ineq:delta-average}.
The variance of our estimator is bounded as follows (the proof is deferred to the end of this section).

\begin{claim}
\label{claim:sample_variance}
Let $\Delta\in \mathcal{H}_q$ be a type such that $\reach_G(\Delta)>\delta$. Then $\Var_\stream[q_\Delta]\le \frac{\E_\stream[q_\Delta]}{s\cdot \cst(q)} + \frac{\sqrt{\alpha q^{2q}\cdot c'_q}}{\cst(q)^2}$.
\end{claim}

Now recall that $\E_\stream[q_\Delta] \ge \reach_G(\Delta) - \frac{\delta}{2} \ge \frac{\delta}{2}$, as we have assumed that $\reach_G(\Delta) \ge \delta$. Let $\eta = \frac{\delta}{2}$, and we apply Chebyshev's inequality to obtain that
\begin{align*}
	\Pr_\stream[|q_\Delta - \E_\stream[q_\Delta]| \ge \eta \E_\stream[q_\Delta]]
        &\le
	\frac{\Var[q_\Delta]}{\eta^2\E_\stream[q_\Delta]^2}
        \\
	   &\le
    \frac{1}{\eta^2}\cdot \frac{\Var[q_\Delta]}{(\E_\stream[q_\Delta])^2}
        \\
	   &\le
	\frac{1}{\eta^2} \left(
	\frac{1}{s\cdot \cst(q) \E_\stream[q_\Delta]} + \frac{\sqrt{\alpha q^{2q}\cdot c'_q}}{\cst(q)^2 \E_\stream[q_\Delta]} \right)
		\\
		&\le
	\frac{8}{s\cdot \cst(q) \delta^3} + \frac{8 \sqrt{\alpha q^{2q}\cdot c'_q}}{\cst(q)^2 \delta^3}	
		\\
        &\le
    \frac{1}{200 |\mathcal{H}_q|}
        \enspace,
\end{align*}
where the last inequality follows from our setting that $\alpha = \frac{\delta^6}{6400|\mathcal{H}_q|^2q^{2q}c'_q} \le  \frac{\cst(q)^4\delta^6}{6400|\mathcal{H}_q|^2q^{2q}c'_q}$ and that $s \ge \max\{\frac{1}{20\sqrt{\alpha q^{2q} \cdot c'_q}},\frac{5000|\mathcal{H}_q|}{\cst(q) \delta^3}\}$.

Therefore, with probability at least $1-\frac{1}{200}$, for all $\Delta\in \mathcal{H}_q$ with $\reach_G(\Delta)\ge \delta$, we have that
\begin{align*}
    q_\Delta
    &\ge (1-\eta)\E_\stream[q_\Delta]
    \ge (1-\frac{\delta}{2})(\reach_G(\Delta)-\frac{\delta}{2})
    \ge \reach_G(\Delta)-\delta
    \fstop
\end{align*}

Finally, with success probability at least $1 - \frac{1}{200} - \frac{1}{200} > \frac{99}{100}$, we have that $q_\Delta \ge \reach_G(\Delta)-\delta$ for all $\Delta \in \mathcal{H}_q$. This finishes the proof of the lemma.
\end{proof}

\subsubsection*{Proofs of \cref{claim:streamcollectvisitedges} and \cref{claim:sample_variance}}

\begin{proof}[Proof of \cref{claim:streamcollectvisitedges}]
We first show that with probability (over the randomness of choosing $u,v$ and $\stream(G)$) at most $2\alpha_0$, there exists some edge that is contained in both $H_u$ and $H_v$. Letting $u,v$ be two vertices that are sampled uniformly at random and $Y := Y_{uv}$, that is
\begin{eqnarray}
\label{eqn:stream_edge2}
	Y = \E_{u,v} [Y_{uv}]
	= \frac{1}{n^2}\cdot \sum_{u,v}Y_{uv}
    \le 2\alpha_0
    \fstop
\end{eqnarray}
	
Let $X$ denote the number of pairs $u,v$ satisfying \cref{eqn:stream_sharing_edge}. Then it holds that
\begin{align*}
    \frac{1}{n^2} \cdot \left((n^2- X) \cdot \sqrt{\alpha_0}\right)
        \le 2 \alpha_0
    \enspace,
\end{align*}
which gives that $\frac{X}{n^2}\ge 1-2\sqrt{\alpha_0}$. This will complete the proof of \cref{claim:streamcollectvisitedges}.

In the following, we sketch the proof of \cref{eqn:stream_edge2}. In order to do so, we only need to modify the corresponding parts in the proof for \cref{thm:random-bfs-visit-edges}. That is, for the given $\alpha>0$, we can define a set $V'_\alpha$ that contains all vertices that can be collected with probability at least $\alpha$ by invoking \textsc{StreamCollect}($\stream(G),v,q$) from a randomly chosen vertex $v$. Then by letting $c'_q:=\sum_{i=0}^{q+1} (q^{2q})^i$, we can show that $|V'_\alpha|\le \frac{c'_q}{\alpha}$ in a similar way as the proof for \cref{thm:number-collisions}, with the only difference that the \textsc{StreamCollect} might see at most $c'_q$ vertices. We can further show that for any $v\in V'_\alpha$, it holds that $\deg(v)\ge n\alpha/c'_q$ as we proved in \cref{thm:collision-high-degree}.
	
Let $S=$\textsc{StreamCollect}($\stream(G),v,q$) for a randomly chosen $v\in V$, and let $S_i\subseteq S$ denote the set of vertices of distance exactly $i$ from the root of $S$. Similar to \cref{ref:claim-who-knows-what-for}, we first prove by induction the following statement:

\begin{itemize}
\item[(*)]\label{statement} For $0\le i\le q+1$, and for every non-isolated vertex $u \in V$, conditioned on the event $\mathcal{E}_{i-1}$ that $u$ is not contained in $\cup_{j\le i-1} S_j$, the probability that $u$ is contained in $S_i$ is at most $\frac{q^{2qi} \deg(u)}{n}$.
\end{itemize}
However, there are some subtle differences in the analysis of the induction from the proof of \cref{ref:claim-who-knows-what-for}. Details follow.
	
First, the statement (*) is true for $i=0$ as we sample $u$ with probability $\frac1n$. Consider the case $i=1$. If we see $u$ conditioned on the event that we did not see it in $S_0$, this implies that one of the neighbors of $u$, say $w \in \ngh{u}$, is in $S_0$, which happens with probability $\frac1n \le \frac{\deg(w)}{n}$. Then, \textsc{StreamCollect} will add the edge $(w,u)$ if it is among the first $q^{2q}$ edges of all $\deg(w)$ edges that are incident to $w$. By the union bound, the probability that $u \in S_1$ conditioned on $\mathcal{E}_i$ is therefore at most $\sum_{w\in \ngh{u}}\frac{\deg(w)}{n} \cdot \frac{q^{2q}}{\deg(w)} = \frac{q^{2q}\deg(u)}{n}$.
Assume the statement holds for $i-1$ and let $i\ge 2$. Then we note that conditioned on $\mathcal{E}_{i-1}$, for any $w\in \ngh{u}$, either $(u,w)$ has already appeared before the time stamp $\tau_w$ that we explore $w$, or $(u,w)$ appeared after we explore $w$. In the former case, we will not see $u$ from $w$. In the latter case, suppose further that there are exactly $j$ edges incident to $w$ that appear after $\tau_w$. Note that since $(u,w)$ appears after $\tau_w$ and $w$ is not in $S_0$, it holds that $j\ge 1$. Then the probability that we will collect $u$ from $w$ is at most $\frac{1}{\deg(w)-j+1}\cdot \frac{q^{2q}}{j}\le \frac{q^{2q}}{\deg(w)}$: with probability $\frac{1}{\deg(w)-j+1}$ the edge $(u,w)$ appears after $\tau_w$ (i.e., after the $\deg(w)-j$ edges that are incident to $w$ and appear before $\tau_w$), and conditioned on the event that $(u,w)$ appears after $\tau_w$, with probability at most $\frac{q^{2q}}{j}$, the edge $(u,w)$ appears among the first $q^{2q}$ edges of all the~$j$ edges incident to $w$ that appear after $\tau_w$. Since $\frac{q^{2q}}{\deg(w)}$ is independent of $j$, the probability that we will collect $u$ from $w$ is at most
\begin{align*}
    \sum_{j=1}^{\deg(w)}
        &
    \left(\Pr[\text{$u$ will be collected from $w$} \mid \text{$j$ neighbors appear after we explore $w$}]\right.
        \\
	   &
    \cdot \left. \Pr[\text{$j$ neighbors appear after we explore $w$}]\right)
        \\
	\le & \frac{q^{2q}}{\deg(w)}
        \fstop
\end{align*}
	
Therefore, by induction on $i-1$, the probability that $u$ is contained in $S_i$ is at most
	
\begin{align*}
	\sum_{\substack{w\in \ngh{u}}} \frac{q^{2q(i-1)}\deg(w)}{n} \frac{q^{2q}}{\deg(w)}
	\le
	\sum_{\substack{w\in \ngh{u}}} \frac{q^{2q}}{n}
	\le
	\frac{q^{2qi} \deg(u)}{n}
	\fstop
\end{align*}
	
This finishes the proof of the statement (*). It follows that the probability that $u$ is contained in $S_i$ is at most $\sum_{i=0}^{q+1} \frac{q^{2qi}\deg(u)}{n}=\frac{c'_q\deg(u)}{n}$. By a similar argument as in \cref{thm:collision-high-degree}, we can then show that for any $v\in V'_\alpha$, it holds that $\deg(v)\ge n\alpha/c'_q$.
	
Finally, we use the same argument for the proof of \cref{thm:random-bfs-visit-edges} to show that with probability $1 - 2 \alpha\cdot q^{2q}\cdot c'_q$, no edge will be contained in both $H_u$ and $H_v$. This finishes the proof of \cref{eqn:stream_edge2}.
\end{proof}

\begin{proof}[Proof of \cref{claim:sample_variance}]
Note that $\Var_\stream[q_\Delta] = \frac{1}{s^2\cst(q)^2} \Var_\stream[\sum_{v\in S}X_v]$. We further have that
\begin{align*}
    \Var_\stream[\sum_{v\in S}X_v]
        &=
    \E_\stream[(\sum_{v\in S}X_v)^2] - (\sum_{v\in S} \E_\stream[X_v])^2
        \\
        &=
	\sum_{v \in S} \E_\stream[X_v^2] +
        \sum_{\substack{u,v \in S \\ u \neq v}} \E_\stream[X_u X_v] -
        \sum_{v\in S} (\E_\stream[X_v])^2 -
        \sum_{\substack{u,v \in S\\ u \neq v}} \E_\stream[X_u]\E_\stream[X_v]
        \\
        &=
    \sum_{v\in S} \E_\stream[X_v] -
        \sum_{v\in S} (\E_\stream[X_v])^2 +
        \sum_{\substack{u,v \in S\\ u \neq v}} \E_\stream[X_u X_v] -
        \sum_{\substack{u,v \in S\\ u \neq v}} \E_\stream[X_u]\E_\stream[X_v]
        \fstop
\end{align*}
	
Let $\mathcal{C}_u, \mathcal{C}_v$ be the set of colored $q$-bounded discs rooted at $u$ and $v$, respectively, that are isomorphic to $\Delta$.
	
Let $u,v \in S$, and consider the execution of \textsc{StreamCollect} on $\stream(G)$ for $u$ and $v$, respectively. To simplify notation, we define the random variables $H_u = (U_u, F_u) \defeq \textsc{StreamCollect}(\stream(G), u, q)$ and $H_v = (U_v, F_v) \defeq \textsc{StreamCollect}(\stream(G), v, q)$. By the definition of expectation,
	
\begin{align*}
    \E_\stream[X_u X_v]
		&=
    \Pr_\stream[X_u = 1 \cap X_v = 1]
        \\
		&=
    \sum_{H \in \mathcal{C}_u} \sum_{H' \in \mathcal{C}_v}
        \Pr_\stream[H_u = H \cap H_v = H']
        \\
		&=
    \sum_{H \in \mathcal{C}_u, H'\in \mathcal{C}_v, E(H)
        \cap E(H')=\emptyset}\Pr_\stream[H_u = H \cap H_v = H']
        \\
		& \,\,\,\, +
    \sum_{H \in \mathcal{C}_u, H'\in \mathcal{C}_v, E(H)\cap E(H')
        \neq \emptyset}\Pr_\stream[H_u = H \cap H_v = H'] %
    \fstop
\end{align*}
	
We further note that for any two $H$ and $H'$, if $E(H) \cap E(H') = \emptyset$, then the events $H_u = H$ and $H_v = H'$ are independent. Therefore, we have that
\begin{align*}
    \sum_{\substack{H \in \mathcal{C}_u, H' \in \mathcal{C}_v\\ E(H) \cap E(H') = \emptyset}} \Pr_\stream[H_u = H \cap H_v = H']
        =
\sum_{\substack{H \in \mathcal{C}_u, H' \in \mathcal{C}_v\\ E(H) \cap E(H') = \emptyset}} \Pr_\stream[H_u = H] \Pr_\stream[H_v = H']
        \le
	\E_\stream[X_u] \cdot \E_\stream[X_v]
	\fstop
\end{align*}
	
Furthermore, we note that by \cref{claim:streamcollectvisitedges},
\begin{align*}
		\sum_{H \in \mathcal{C}_u, H'\in \mathcal{C}_v, E(H)\cap E(H')\neq \emptyset}\Pr_\stream[H_u = H \cap H_v = H']
		=
		\Pr_\stream[\text{$H_u$ and $H_v$ visited the same edge}]
		\le
		\sqrt{\alpha q^{2q}\cdot c'_q}
		\enspace,
\end{align*}
where the last inequality follows from our condition that event $\mathcal{E}$ holds.
	
Thus, we have that for any two $u,v\in S$,
\begin{align*}
	\E_\stream[X_uX_v]
        &\le
    \E_\stream[X_u]\E_\stream[X_v] + \sqrt{\alpha q^{2q}\cdot c'_q}
	\fstop
\end{align*}
	
This implies that
\begin{align*}
	\Var[\sum_{v\in S}X_v]
        &=
    \sum_{v\in S} \E_\stream[X_v] - \sum_{v\in S} (\E_\stream[X_v])^2
        + \sum_{\substack{u,v \in S\\ u \neq v}} \E_\stream[X_u X_v]
        - \sum_{\substack{u,v \in S\\ u \neq v}}
        \E_\stream[X_u]\E_\stream[X_v]
        \\
		&\le
    \sum_{v\in S} \E_\stream[X_v] - \sum_{v\in S} (\E_\stream[X_v])^2
        + s^2 \sqrt{\alpha q^{2q}\cdot c'_q}
        \\
		&\le
    \sum_{v\in S} \E_\stream[X_v] + s^2 \sqrt{\alpha q^{2q}\cdot c'_q}
	\fstop
\end{align*}
	
Thus,
\begin{align*}
    \Var[q_\Delta]
        &\le
    \frac{\sum_{v\in S} \E_\stream[X_v] +
            s^2 \sqrt{\alpha q^{2q}\cdot c'_q}}{s^2 \cst(q)^2}
        \le
    \frac{\E_\stream[q_\Delta]}{s\cdot \cst(q)} +
        \frac{\sqrt{\alpha q^{2q}\cdot c'_q}}{\cst(q)^2}
	\fstop
	\qedhere
	\end{align*}
\end{proof}

\section{Testing Graph Properties in Random Order Streams}
\label{sec:testing-in-streams}

Now we transform constant-query property testers (with one-sided error) into constant-space streaming property testers, and prove \cref{thm:stream-test}. %
The main idea is to explore the streamed graph by \textsc{StreamCollect} and look for the forbidden subgraphs in $\mathcal{F}_n$ that characterize $\Pi$ (see \cref{thm:canonical-tester}). However, in the underlying analysis, we use the (reversible) decomposition of the forbidden subgraphs in $\mathcal{F}_n$ into $\mathcal{F}'_n$ (see \cref{thm:refined_characterization}) to prove the following: if $\mathcal{T}$ finds the colored $q$-bounded discs $\Delta_1, \ldots, \Delta_k$ that compose a forbidden subgraph $F \in \mathcal{F}'_n$ with probability $p$, then the streaming tester will find at least as many copies of $\Delta_1, \ldots, \Delta_k$ as $\mathcal{T}$ (see \cref{lemma:lower_bound_disc_probability}) and can stitch $F$ from these copies. With these tools at hand, we can incorporate our analysis from previous sections to complete the proof of \cref{thm:stream-test}.

\begin{proof}[Proof of \cref{thm:stream-test}]
We let $q_0=q_0(\varepsilon)$ denote the query complexity of $\Pi$. Let $n=|V|$. We present our testing algorithm. Let $q=c\cdot q_0$ for some constant $c$ from Theorem~\ref{thm:canonical-tester}. Let $\alpha = \frac{\delta^6}{6400|\mathcal{H}_q|^2q^{2q}c'_q}$, where $c'_q=\sum_{i=0}^{q+1}q^{2qi}$, and $\delta = \frac{1}{200 \lvert \mathcal{H}_q \rvert}$. If $n\le n_0:= \frac{qc_q}{\alpha^2}$, then we simply store the whole graph. If $n> n_0$, we proceed as follows. Let $\mathcal{F}_n$ be the set of forbidden subgraphs that characterize $\Pi$ as stated in \cref{thm:canonical-tester}. We sample $s \ge \max\{\frac{1}{20\sqrt{\alpha q^{2q}\cdot c'_q}},\frac{5000|\mathcal{H}_q|}{\cst(q) \delta^3} \} $ vertices $S \subseteq V$ and run $\textsc{StreamCollect}(\stream(G), v, q)$ for each $v \in S$ to obtain a subgraph $H_v = (V_v, E_v)$ of $G$. If $H = \cup_{v \in S} H_v$ contains a forbidden subgraph $F \in \mathcal{F}_n$, the tester rejects, otherwise it accepts. See \cref{alg:stream-test} for details.

\begin{algorithm}[H]
	\caption{Testing graph property $\Pi$ in random order stream}
	\label{alg:stream-test}
	\begin{algorithmic}
		\Function{StreamTest}{$\stream(G), n,\varepsilon,\mathcal{F}_n${}}
		\State $S \gets \text{sample s vertices u.a.r. from } V$
		\ForAll{$v \in S$}
			\State $H_v \gets (V_v, E_v) = \textsc{StreamCollect}(\stream(G),v,q)$
		\EndFor
		\State $H \gets (\cup_v V_v, \cup_v E_v)$
		\If{there exists $F\in \mathcal{F}_n$ such that $H$ contains a subgraph $F$}
		\State{Output \textbf{Reject}}
		\Else
		\State{Output \textbf{Accept}}
		\EndIf
		\EndFunction
	\end{algorithmic}
\end{algorithm}
	
The space complexity of the algorithm is $s\cdot q_0^{O(q_0)} = O_{q_0}(1)$ words. For the correctness of the algorithm, we note that for any property $\Pi$ that is constant-query testable with one-sided error, then with probability $1$, we will not see any $F\in \mathcal{F}'_n$ if the graph $G$ satisfies $\Pi$.

On the other hand, if $G$ is $\varepsilon$-far from satisfying $\Pi$, then by \cref{thm:refined_characterization}, with probability at least $\frac{2}{3}$, the subgraph $S_q$ spanned by the union of $q$-bounded discs rooted at $q$ uniformly sampled vertices from $G$ will span a subgraph that is isomorphic to some $F\in \mathcal{F}'_n$. Note that, in contrast to the algorithm above, the analysis uses the decomposition of forbidden subgraphs in $\mathcal{F}_n$ into colored $q$-discs given by \cref{thm:refined_characterization}. The key idea is to use the $q$-bounded discs that \textsc{StreamCollect} and the implicit colors (which are not observed by \textsc{StreamCollect}, but can be used in the analysis to identify vertices in $V_\alpha$) to stitch forbidden subgraphs from $\mathcal{F}'_n$ that are discovered by \textsc{RandomBFS}. We prove that with sufficient probability, for each colored $q$-bounded disc $\Delta$, \textsc{StreamCollect} finds at least as many copies of $\Delta$  as \textsc{RandomBFS}, and therefore, it can reproduce the same types of forbidden subgraphs from $\mathcal{F}'_n$.

By Markov's inequality and the union bound, the probability that at least one $q$-RBFS in the canonical tester for $\Pi$ will return a colored $q$-bounded disc that is isomorphic to a disc $\Delta'$ such that $\reach_{G}(\Delta') < 2 \delta=\frac{1}{100 |\mathcal{H}_q|}$ is at most $\frac{1}{100}$. Let $\mathcal{D}$ be the set of all colored $q$-bounded discs $\Delta$ such that $\reach_{G}(\Delta) \ge 2 \delta$. %

By \cref{lemma:lower_bound_disc_probability}, with probability at least $1-\frac{1}{100}$, for every $\Delta \in \mathcal{D}$, the number of graphs $H_v$ obtained by $\textsc{StreamCollect}$ contains a subgraph isomorphic to $\Delta$ is at least $100 \lvert \mathcal{H}_q \rvert \cdot \reach_{G}(\Delta)\ge 1$. By (implicitly) coloring all vertices in $V_\alpha$, it follows from \cref{thm:refined_characterization} that $H$ contains a forbidden subgraph from $\mathcal{F}'_n$ with probability $1 - \frac{1}{100} - \frac{1}{100} > \frac{2}{3}$.
\end{proof}

\section{Extension to Random Neighbor/Edge Model}
\label{sec:other-models}

While in the random neighbor model one can easily access a random vertex of the input graph, it is not necessarily easy to access a random edge. The lack of this feature is especially unhelpful in the context of general graphs, without bounds for the maximum degree. Indeed, for bounded degree graphs with maximum degree $d$, by choosing a random vertex and then its random neighbor, every edge will be chosen with probability in the interval $[\frac{1}{dn},\frac1n]$, which is close to $\frac1{|E|}$ assuming $|E| = \Omega(n)$. However, for general graphs, we do not have any similar relationship. For example, one could exploit this feature to obtain a lower bound of $\Omega(\sqrt{n})$ for testing planarity in general graphs in the random neighbor model, since in this model it is impossible to distinguish with $o(\sqrt{n})$ queries (in expectation) between a trivial planar graph that is edge-less graph and a highly non-planar graph that has $n-\sqrt{n}$ isolated vertices and a clique of size $\sqrt{n}$ on the other vertices.
To avoid uninteresting cases, we assume that $|E| = \Omega(|V|)$ in this section.

\begin{definition}[\textbf{Random neighbor/edge model}]
\label{def:random-bfs-edge}
In the \emph{random neighbor/edge model}, an algorithm is given $n \in \setn$ and access to an input graph $G = (V,E)$ by a query oracle, where $V = [n]$. The algorithm may ask a query based on all knowledge it has gained by the answers to previous queries. There are two types of queries in the random neighbor/edge model:
\begin{itemize}
\item \textbf{random edge query:} the oracle returns an edge $e \in E$ chosen i.u.r. from $E$;
\item \textbf{random neighbor query} takes as its input a specified vertex $v \in V$, and the oracle returns a vertex that is chosen i.u.r. from the set of all neighbors of $v$.
\end{itemize}
\end{definition}

It turns out that our previous results on canonization of constant-query testers in random neighbor model and the emulation of such algorithms in random order streaming model can be easily extended to the random neighbor/edge model. As most of the proof ideas are the same as before, we will only give informal discussions of the results in this section and describe the difference from the previous results/proofs.

For example, we can give the canonical tester in random neighbor/edge model in the following theorem, whose proof is almost the same as the one for \cref{thm:canonical-tester}, with the only exception that now the canonical tester first samples both $q'$ vertices and $q'$ edges i.u.r.

\begin{theorem}[\textbf{Canonical tester in random neighbor/edge model}]
\label{thm:canonical-tester-edge-model}
Let $\Pi = (\Pi_n)_{n \in \setn}$ be a graph property that can be tested in the random neighbor/edge model with query complexity $q = q(\varepsilon,n)$ and error probability at most $\frac{1}{3}$. Then for every $\varepsilon$, there exists an infinite sequence $\mathcal{F} = (\mathcal{F}_{n})_{n \in \setn}$ such that for every $n \in \setn$,
\begin{itemize}
\item $\mathcal{F}_{n}$ is a set of rooted graphs such that each graph $F \in \mathcal{F}_{n}$ is the union of $q'$ many $q'$-bounded discs;
\item the property $\Pi_n$ on $n$-vertex graphs can be tested with error probability at most $\frac13$ by the following canonical tester:
    \begin{enumerate}
    \item sample $q'$ vertices i.u.r. and mark them \emph{blue roots};
    \item sample $q'$ edges i.u.r. and mark their endpoints \emph{red roots};
    \item for each root $v$ (either a blue or red root), perform a $q'$-RBFS starting at $v$;
\item reject if and only if the explored subgraph is root-preserving isomorphic to some $F \in \mathcal{F}_{n}$,
    \end{enumerate}
\end{itemize}
where $q'=cq$ for some constant $c > 1$. The query complexity of the canonical tester is $q^{O(q)}$. Furthermore, if $\Pi = (\Pi_n)_{n \in \setn}$ can be tested in the random neighbor model with one-sided error, then the resulting canonical tester for $\Pi$ has one-sided error too, i.e., the tester always accepts graphs satisfying $\Pi$.
\end{theorem}

We can also have a refined canonical tester for the random neighbor/edge model that uses colors to distinguish intersecting vertices, which corresponds to the one from  \cref{thm:refined_characterization} for the random neighbor model. Again, such a canonization can be proved in a similar way as before, with the only difference that we also need to define the reach probability of a vertex corresponding to a $q$-RBFS from a uniformly sampled edge. This difference can be easily handled by viewing the process of sampling an edge as the (equivalent) process of sampling a vertex with probability proportional to its degree.

Finally, noting that a set of $q'$ random edges can be easily obtained by taking the first $\Theta(q')$ edges in the randomly ordered edge stream, we can then use the above refined canonical tester and the previous emulation proof in \cref{sec4} to transform all the testers in the random neighbor/edge query model to the random order streaming model.

\begin{theorem}\label{thm:stream-test-random-edge}
Every graph property $\Pi$ that is constant-query testable with one-sided error in the random neighbor/edge model is also constant-space testable with one-sided error in the random order graph streams.
\end{theorem}

\section{Conclusions}

We gave the first canonical testers for all constant-query testers in the random neighbor model for general graphs and show that one can emulate any constant-query tester with one-sided error in this query model in the random order streaming model with constant space. Our transformation between constant-query testers and streaming algorithms with constant space provides a strong and formal evidence that property testing and streaming algorithms are very closely related. Our results also work for any restricted class of general graphs and other query models, e.g., random neighbor/edge model. It follows that many properties are constant-space testable (with one-sided error) in random order streams, including $(s,t)$-disconnectivity, being $d$-bounded degree, $k$-path-freeness of general graphs and bipartiteness and $H$-freeness of planar (or minor-free) graphs.

It will be very interesting to transform all constant-query testers with \emph{two-sided errors} in the above mentioned query models to constant-space testers in the random order streaming model. Such a result is not possible with the current techniques that only detect a forbidden subgraph (which works for testers with one-sided error), while are not able to approximate the frequencies of all forbidden subgraphs.

\bibliographystyle{alpha}
\bibliography{literature}

\end{document}